%% file: paper.tex
\documentclass[preprint,9pt, numbers,nocopyrightspace]{sigplanconf}

\usepackage[scaled]{helvet} 
\usepackage{url}                  
\usepackage{listings}          
\usepackage{amsmath,amsthm,amssymb}
\usepackage{stmaryrd}
\usepackage{graphicx}
\usepackage{caption}
\usepackage{subcaption}

\usepackage[T1]{fontenc}

\usepackage{xspace}
\usepackage{tabularx}
\usepackage{listliketab}
\usepackage{slantsc}
\usepackage{lipsum} 
\usepackage[colorlinks=true,allcolors=blue,breaklinks,draft=false]{hyperref}

\usepackage{algorithm}
\usepackage{algpseudocode}

\usepackage[show]{ed}
\ifshowednotes
\RequirePackage{paralist}
\RequirePackage{xcolor}
\fi

\newif\iffull
\fulltrue

\newcommand{\doi}[1]{doi:~\href{http://dx.doi.org/#1}{\Hurl{#1}}}   

\newtheorem{theorem}{Theorem}
\newtheorem{lemma}[theorem]{Lemma}
\newtheorem{definition}[theorem]{Definition}
\newtheorem{proposition}[theorem]{Proposition}

\input{macros}

\begin{document}

\title{\textsc{NoFAQ}:  Synthesizing Command Repairs from Examples}

%
%


\authorinfo{Loris D'Antoni}
{\makebox{University of Wisconsin-Madison}}
{loris@cs.wisc.edu}

\authorinfo{Rishabh Singh}
{\makebox{Microsoft Research}}
{risin@microsoft.com}

\authorinfo{Michael Vaughn}
{\makebox{University of Wisconsin-Madison}}
{vaughn@cs.wisc.edu}

\maketitle

\input{abstract}

\input{introduction}
\input{motivatingexamples}
\input{language}
\input{synthesisalgorithm}
\input{evaluation}

\input{limitations}

\input{related}
\input{conclusion}

\bibliographystyle{abbrvnat}
\bibliography{references}

\iffull
\appendix
\input{appendix}
\fi

\end{document}

%% file: macros.tex
\DeclareCaptionFormat{myformat1}{\hrulefill\newline #1#2#3}
\DeclareCaptionFormat{myformat2}{#1#2#3}
\newcommand{\fcomm}[1]{\vspace{2mm}\State \textcolor[rgb]{0.00,0.30,0.00}{{//{#1}}}}

\algtext*{EndWhile}
\algtext*{EndIf}
\algtext*{EndFor}
\algtext*{EndFunction}
\floatname{function}{fun}
\algnewcommand\algorithmicswitch{\textbf{switch}}
\algnewcommand\algorithmiccase{\textbf{case}}
\algnewcommand\algorithmicassert{\texttt{assert}}
\algnewcommand\Assert[1]{\State \algorithmicassert(#1)}%
\algdef{SE}[SWITCH]{Switch}{EndSwitch}[1]{\algorithmicswitch\ #1\ \algorithmicdo}{\algorithmicend\ \algorithmicswitch}%
\algdef{SE}[CASE]{Case}{EndCase}[1]{\algorithmiccase\ #1}{\algorithmicend\ \algorithmiccase}%
\algtext*{EndSwitch}%
\algtext*{EndCase}%

\def\thefaq{\textsc{NoFAQ}\xspace}
\def\thefuck{\textsc{TheFxxx}\xspace}

\newcommand{\figlabel}[1]{\label{fig:#1}}
\newcommand{\figref}[1]{Figure~\ref{fig:#1}}

\renewcommand{\t}[1]{\textbf{#1}}
\newcommand{\tconst}[1]{\textsc{#1}}
\renewcommand{\tt}[1]{\texttt{#1}}
\newcommand{\bigO}[1]{O(#1)}

\def\e{m}
\def\g{f}
\newcommand{\langc}{\textsc{Fixit}\xspace}
\newcommand{\bnfalt}{\;|\;}
\newcommand{\constStr}{\tconst{Str}}
\newcommand{\varmatch}{\tconst{Var-Match}} 
\newcommand{\varexpr}{\tconst{Var}}
\newcommand{\tmatch}{\t{match}}
\newcommand{\tand}{\t{and}}
\newcommand{\tapply}{\ensuremath{\rightarrow}}
\newcommand{\tsapply}{\ensuremath{\tapply_s}}
\newcommand{\fixRule}[3]{\tmatch~#1~\tand~#2~\tapply~#3}
\newcommand{\sfixRule}[3]{\tmatch~#1~\tand~#2~\tsapply~#3}
\newcommand{\origCmd}{cmd}
\newcommand{\errMsg}{err}
\newcommand{\fixCmd}{fix}
\newcommand{\fconstStr}{\tconst{Fstr}}
\newcommand{\substrAppend}{\tconst{Sub-lr}}
\newcommand{\constPos}{\tconst{Ipos}}
\newcommand{\symPos}{\tconst{Cpos}}
\newcommand{\pl}{p_L}
\newcommand{\pr}{p_R}
\newcommand{\pref}{l}
\newcommand{\suf}{r}
\newcommand{\off}{\delta}
\newcommand{\strtt}[1]{\constStr(\tt{#1})}

\newcommand{\state}{\sigma}
\newcommand{\semantics}[1]{\llbracket{#1}\rrbracket}
\newcommand{\semanticsfix}[1]{\semantics{#1}}
\newcommand{\semanticsfun}[1]{\semantics{#1}^\t{fun}}
\newcommand{\semanticspos}[1]{\semantics{#1}^\t{pos}}

\usepackage{yfonts}
\newcommand{\tsc}[1]{\texttt{#1}}
\def\unify{\tsc{unify}}
\def\substr{\tsc{substr}}
\def\matchexpr{\tsc{match-expr}}

\def\size{\tsc{size}}
\def\indices{\tsc{indices}}

\newcommand\length[1]{|#1|}
\def\conc{\tsc{con}}
\def\ind{I}

\def\scmd{\bar{s}_{cmd}}
\def\serr{\bar{s}_{err}}
\def\sfix{\bar{s}_{fix}}
\def\symfix{\ensuremath{fixes}} 

\theoremstyle{remark}
\newtheorem*{proofsketch}{Proof sketch}

%% file: abstract.tex
\begin{abstract}
Command-line tools are confusing and hard to use
for novice programmers due to their cryptic error messages and lack
of documentation. Novice users often resort to online help-forums for finding corrections to their buggy commands, but have a hard time in searching precisely for posts that are relevant
to their problem and then applying the suggested solutions to their buggy command.

We present a tool \thefaq that uses a set of rules to suggest possible fixes
when users write buggy commands that trigger commonly occurring errors.
The rules are expressed in a language called \langc and each
rule pattern-matches against the user's buggy command and the corresponding error message, and uses
these inputs to
produce a possible fixed command.
Our main contribution is an algorithm based on lazy VSA for synthesizing \langc rules from
examples
of buggy and repaired commands. The algorithm allows users to add new rules
in \thefaq without having to manually encode them.
We present the evaluation of \thefaq on
92 benchmark problems and show that \thefaq is able to instantly synthesize
rules for 81 benchmark problems in real time using just 2 to 5 input-output examples for each rule.
\end{abstract}

%% file: introduction.tex
\section{Introduction}

Command-Line Interfaces (CLI) let users interact with a computing system
by writing sequences of commands.
CLIs are especially popular
amongst advanced computer users, who
use them to perform small routine tasks such as
committing a file to a repository with version control,
installing software packages,
compiling source code,
finding and searching for files etc.
Even though this mode of interaction is getting replaced by more natural graphical user interfaces,
CLIs are still routinely used for most scripting tasks in Unix and Mac OS.
Even the Windows operating system now officially provides complex command-line interfaces with
products such as Windows Powershell.

Since command-line interactions often require complex parameters and flag settings for specifying the desired intent,
non-expert users find CLIs challenging to use~\cite{promptcomprehension}.
Moreover, after entering an incorrect input command, the user has to deal with cryptic errors that are hard to decipher by just
looking at the verbose text-based documentation of the commands.
For these reasons, users typically resort to online help-forums for finding corrections to their buggy commands.
Unfortunately, this can also be problematic
as users need to precisely search for posts that relate to the issues with their commands and then transform the suggested solutions
to apply them in their context.
Sometimes users also need to create a new post if no relevant post exists (or can be found), and then need to wait for hours or days to obtain a solution to their problem.

\paragraph{What about common errors?}
Recently, a tool \thefuck\footnote{We decided to censor the name of the tool. The tool can be found at \url{http://bit.ly/CmdCorrection}.} was developed for automatically addressing common errors when working with a CLI.
If after typing a command a user receives an error message,
\thefuck uses a set of hard-coded rules to suggest possible fixes to
the user's command. Each rule pattern-matches against the input command
and the error message, and uses these inputs to produce a possible fixed command.
Typical fixes include adding missing flags, creating a missing directory, or changing file extensions.
\thefuck has become
extremely popular and, on GitHub, it
has already been starred by more than 20,000 users and has been forked more than 1000 times.
Despite its success, the tool also has a main limitation:
to add a new rule a developer first needs to
understand the syntax and precise semantics of \thefuck  and
then manually hard-code the rule into the tool.
Due to this complexity, newly added rules have at times caused non-terminating or
unexpected behaviors\footnote{\url{http://bit.ly/1j7zxOr} and \url{http://bit.ly/1YgngXJ}.}.

\paragraph{Synthesizing repairs from examples}
Inspired by the success and limitations of \thefuck, we built
\thefaq (No more Frequently Asked Questions), a tool for automatically
addressing common errors in CLIs.
\thefaq also uses a set of rules for fixing common errors,
but it differs from \thefuck in the following key aspects:
\begin{enumerate}
\item Rules are encoded in a declarative domain-specific language (DSL)  called $\langc$.
\item To add new rules,
 users only provide examples of buggy and repaired commands
 and \thefaq automatically synthesizes the desired \langc rules that are consistent with the given examples.
\end{enumerate}

We envision \thefaq system being used by non-developers and end-users, who can easily extend the system by providing new examples of fixes. The long-term goal of the system is to learn from examples obtained from shell histories of thousands of users in an unsupervised manner.  Although a developer can write similar rules in a system like \thefuck manually, there are two main challenges with doing so. First, it is not feasible to easily add thousands of rules as end-users generally do not have contributor access to \thefuck's source code. Second, even for developers, writing correct rules is difficult and error-prone because of the complexity of the string manipulations needed to perform the command fixes. In fact, \thefuck only consists of less than 100 rules in a little over 1 year, since adding new ones is a fairly complex task. 

The $\langc$ DSL for encoding fix rules is inspired by the types of rules appearing in \thefuck and by common command repairs requested by users on help-forums.
A \langc rule first uses pattern matching and unification to match the command and error message, and then applies a fix transformation if the match succeeds.
The transformations consist of substring and append functions on strings present in the command and error message.

We present an algorithm that efficiently synthesizes $\langc$ rules that are consistent with a given set of input-output examples
using a  Version-space Algebra (VSA). VSA-based synthesis techniques are used to
succinctly represent the set of all expressions that are consistent with a set of examples~\cite{cacm12,smartedit}.
Even though existing VSA data structures can represent an exponential number of \langc rules in polynomial space,
this space can still be quite large.
To address this problem we introduce \emph{lazy version-space algebra}.
Given a set of examples, our algorithm maintains a lazy representation 
of only a subset of the \langc rules that are consistent with the examples.
The rules missing from the version-space are only enumerated when necessary---i.e., when a new input-output
example can only be accounted by adding a \langc rule that is not already present in the version-space. 
Because of the careful design of $\langc$, our synthesis algorithm has a polynomial time complexity. In contrast, existing VSA-based synthesis techniques for string transformations require exponential time~\cite{cacm12}. The polynomial time complexity is crucial for our synthesis algorithm to scale to a large number of fix examples.
Since different examples may refer to different target rules, we propose a strategy
to partition the input examples into groups of examples corresponding to individual rules. We then use the lazy VSA 
algorithm to learn the $\langc$ rules for each partition.

We evaluate the synthesis algorithm implemented in $\thefaq$ on 92 benchmark problems obtained from both $\thefuck$ (76) and online help-forums (16). \thefaq is able to learn the repair rules for 81 of the
buggy commands in these benchmark problems from only
2 to 5 input-output examples each.

\paragraph{Contributions summary:}
\begin{enumerate}
\item \langc, a declarative domain-specific language for encoding
rules that map a command
and an error message to possible fixed commands (\S~\ref{sec:language}).
\item A sound and complete polynomial time synthesis algorithm based on lazy version-space algebra for synthesizing \langc rules from input-output examples (\S~\ref{sec:synthalgo}).

\item An analysis of the formal properties of the \langc language and its synthesis algorithm (\S~\ref{sec:algoproperties}).
\item A qualitative and quantitative evaluation of the synthesis algorithm on 92 benchmarks obtained from both $\thefuck$ and online help-forums (\S~\ref{sec:evaluation}).
\end{enumerate} 

%% file: motivatingexamples.tex
\section{Motivating examples}
We first present the main ideas behind \thefaq using some concrete examples.
All the example rules presented in this section are actual rules appearing in \thefuck system.

\subsection{Adding missing file extension}
\label{sec:motiv:adding}
Java programmers, in particular novice ones, are likely to encounter this error 
when they accidentally pass a class name instead of a source code file to the \texttt{javac} compiler.
\begin{table}[H]
\begin{tabularx}{\columnwidth}{|@{\hspace{.5em}}>{\bfseries}l@{\hspace{.5em}}X@{\hspace{.5em}}|}
\hline
cmd1:  & javac Employee \\
err1: & Class names, `Employee', are only accepted if annotation
       processing is explicitly requested\\
\hline
\end{tabularx}
\end{table}

\noindent Often, this error is provided by the \texttt{javac} compiler when it is invoked on a file
that does not have the proper \texttt{.java} extension.
A seasoned programmer would immediately recognize the problem and add the extension \texttt{.java} at the end of the input file.
\begin{table}[H]
\begin{tabularx}{\columnwidth}{|@{\hspace{.5em}}>{\bfseries}l@{\hspace{.5em}}X@{\hspace{.5em}}|}
\hline
fix1: & javac Employee.java\\
\hline
\end{tabularx}
\end{table}

On the other hand,
a novice programmer will search the web in the hope
of finding a way to address the error and understand how to apply it to their setting.
The goal of \thefaq is to automatically synthesize simple fixes like this one from
examples provided by experienced users and
use the synthesized fixes to help
novice users who encounter similar errors.
For example, let's say that a skilled developer provides another 
triple of the following form.

\begin{table}[H]
\begin{tabularx}{\columnwidth}{|@{\hspace{.5em}}>{\bfseries}l@{\hspace{.5em}}X@{\hspace{.5em}}|}
\hline
cmd2:  & javac Pair \\
err2: & Class names, `Pair', are only accepted if annotation
       processing is explicitly requested\\
fix2: & javac Pair.java\\
\hline
\end{tabularx}
\end{table}
Using these two examples \thefaq will synthesize the following fix rule.

\begin{table}[H]
\begin{tabular}{l}
\tmatch~[\strtt{javac}, $\varmatch(1, \varepsilon, \varepsilon)$]\\
\t{and}~[\strtt{Class}, \strtt{names,}, $\varmatch(2, \texttt{`}, \texttt{',})$\\ 
\qquad~ \strtt{are}, \strtt{only}, \strtt{accepted}, \strtt{if}, \\
\qquad~ \strtt{annotation}, \strtt{processing}, \strtt{is}, \\
\qquad~ \strtt{explicitly}, \strtt{requested}]\\
\tapply~
[\fconstStr(\texttt{javac}),
$\substrAppend(0,0,\varepsilon,\texttt{.java},\varexpr(1))]$
\end{tabular}
\end{table}
\noindent
The first part of the rule (i.e., up to the symbol $\tapply$) pattern-matches
against the command and the error message and binds the input strings
to the corresponding variables.
The variables are then used by the second part of the rule to produce the output.
In this case the $\substrAppend(0,0,\varepsilon,\tt{.java},\varexpr(1))$ expression extracts the complete string associated with $\varexpr(1)$ (a start index of $0$  and an end index of $0$ denotes the identity string extraction), and then prepends the string $\varepsilon$ at the beginning of it,
and appends the string \texttt{.java} at the end of it.

\subsection{Extracting substrings}
\label{sec:motiv:extractingsubs}
The following scenario is another common one for novice Java programmers.

\begin{table}[H]
\begin{tabularx}{\columnwidth}{|@{\hspace{.5em}}>{\bfseries}l@{\hspace{.5em}}X@{\hspace{.5em}}|}
\hline
cmd1:  & java Run.java \\
err1: & Could not find or load main class Run.java\\
fix1: & java Run\\
\hline
\end{tabularx}
\end{table}

\noindent Given this example and another similar one, \thefaq synthesizes the following rule.

\begin{table}[H]
\begin{tabular}{l}
\tmatch~[\strtt{java}, $\varmatch(1, \varepsilon, \texttt{.java})$]\\
\t{and}~[\strtt{Could}, \strtt{not}, \strtt{find}, \strtt{or}, \\
  \qquad~ \strtt{load}, \strtt{main}, \strtt{class}, \\
  \qquad~ $\varmatch(2, \varepsilon, \texttt{.java})$]\\
\tapply~
[\fconstStr(\texttt{javac}), $\substrAppend(0,-5,\varepsilon,\varepsilon,\varexpr(1))]$
\end{tabular}
\end{table}

\noindent This rule extracts the substring of the input file name
starting at index $0$ and ending at
the index $-5$ ($5^\text{th}$ index from the end of the string) in order to remove the \tt{.java} extension.

\subsection{Extracting complex substrings}
\label{complexsubs}
A user was trying to move a picture from one location to another but got the following error message. 
\begin{table}[!h]
\begin{tabularx}{\columnwidth}{|@{\hspace{.5em}}>{\bfseries}l@{\hspace{.5em}}X@{\hspace{.5em}}|}
\hline
cmd1:  & mv photo.jpg Mary/summer12.jpg \\
err1: & can't rename 'photo.jpg': No such file or directory\\
fix1: & mkdir Mary \texttt{\&\&} mv photo.jpg Mary/summer12.jpg\\
\hline
\end{tabularx}
\end{table}

\noindent The error is cryptic for novice command-line users and does not guide them towards the actual problem of the missing directory.
Given this example and another similar one, \thefaq synthesizes the following rule.

\begin{table}[H]
\begin{tabular}{l}
  \tmatch~ [\strtt{mv}, \varmatch(1, $\varepsilon$, $\varepsilon$), \\
    \qquad~ $\varmatch(2, \varepsilon, \varepsilon)$]\\
\t{and}~ [\strtt{can't}, \strtt{rename}, $\varmatch(3,\texttt{`},\texttt{'})$, \strtt{No},  \\
    \qquad~ \strtt{such}, \strtt{file}, \strtt{or}, \strtt{directory}]\\
\tapply~ [\fconstStr(\texttt{mkdir}),
$\substrAppend(0,\symPos(\texttt{/},1,0),\varepsilon,\texttt{\&\&},\varexpr(2))$\\
\qquad \fconstStr(\texttt{mv}), $\substrAppend(0,0,\varepsilon,\varepsilon,\varexpr(1))$,\\
\qquad $\substrAppend(0,0,\varepsilon,\varepsilon,\varexpr(2))]$
\end{tabular}
\end{table}

\indent The second expression in the output extracts the directory name---i.e.,
the substring that starts at index 0 and ends at the index of first occurrence of the character \texttt{/}.
The rule also adds a string \texttt{\&\&} at the end of the extracted string to pipe the two output commands. 

%% file: language.tex
\section{The command repair language $\langc$}
\label{sec:language}
We now describe $\langc$, a domain-specific language for expressing repair rules.
The syntax and semantics of $\langc$ is presented in \figref{dslsyntax} and \figref{dslsemantics} respectively.
The language $\langc$ is designed to be expressive enough to capture most of the 
rules we found in \thefuck and in online help-forums,
but at the same time concise enough to enable efficient learning from examples.

\paragraph{General structure}
Each $\langc$ program is a rule of the form
$$\fixRule{\origCmd}{\errMsg}{\fixCmd}$$
that takes
as input a command $\bar{s}_{cmd}$ and an error $\bar{s}_{err}$ and either produces a fixed command or the undefined value $\bot$.
The inputs $\bar{s}_{cmd}$ and $\bar{s}_{err}$ are lists of strings that
are obtained by extracting all the space-separated strings appearing in the input command and error message respectively.
The output fix produced by the rule is also a list of strings.
From now on, we assume that the inputs and outputs are lists of strings that do not contain space characters.

A rule has three components.
\begin{enumerate}
\item A list of match expressions  $\origCmd=[\e_1,\cdots, \e_l]$ used to pattern match against the input command $\bar{s}_{cmd}$.
\item A list of match expressions $\errMsg=[\e_1,\cdots, \e_k]$ used to pattern match against the input error message $\bar{s}_{err}$.
\item A list of fix expressions  $\fixCmd=[\g_1,\cdots, \g_m]$ used to produce the new fixed command.
\end{enumerate}

\paragraph{Match expressions}
A match expression $\e$ is either of the form $\constStr(s)$ denoting a constant string $s$ or of the form
\linebreak $\varmatch(i,l,r)$. 
A $\varmatch(i,l,r)$ expression denotes a variable index $i$ and requires the matched
string to start with the prefix $l$ and end with the suffix $r$. We assume that no two variable expressions appearing in the match expression have the same index.
When a list of match expressions $[\e_1,\cdots, \e_l]$ is applied to a list
of strings $\bar{s}=[s_1,\ldots, s_l]$ with the same length $l$,
it generates a partial function $\state:  \mathbb{N}\mapsto \Sigma^*$
that assigns variables appearing in the match expressions to the corresponding strings in the input.
For example, evaluating 
$$[\constStr(\texttt{mv}),\varmatch(1, \epsilon, \texttt{.jpg}),\varmatch(2, \epsilon, \texttt{.jpg})]$$ 
on the
list of strings $[\texttt{mv},\texttt{a.jpg},\texttt{b.jpg}]$ produces the function
$\state$ such that $\state(1)=\texttt{a.jpg}$ and  $\state(2)=\texttt{a.jpg}$. 
On the other hand, evaluating it on $[\texttt{mv},\texttt{a.png},\texttt{b.jpg}]$ yields $\bot$, as $\texttt{a.png}$ does not match the required suffix in $\varmatch(1, \epsilon, \texttt{.jpg})$.

\begin{figure}[!t]
    $\begin{array}{lrcl}
    \text{Fix rule} & r & := & \fixRule{\origCmd}{\errMsg}{\fixCmd}\\
    \text{Input cmd}& \origCmd & := & [\e_1,\cdots, \e_l] \\
    \text{Input error}& \errMsg & := & [\e_1,\cdots, \e_m] \\
    \text{Output cmd}& \fixCmd & := &  [\g_1,\cdots, \g_n] \\
    \text{Match expr} & \e & := & \constStr(s) \\
                      & & \bnfalt & \varmatch(i, \pref, \suf)\\
    \text{Fix expr} & \g & := & \fconstStr(s)\\
                    & & \bnfalt & \substrAppend(\pl,\pr,\pref,\suf,\varexpr(i))\\
    \text{Pos expr}& p & := & \constPos(k) \\
                        & & \bnfalt &  \symPos(c,k,\off)\\
    \multicolumn{3}{l}{\textrm{Variables and constants:}} \\
    \multicolumn{4}{l}{
    \begin{array}{rlrl}
    s,s_l,s_r:& string & i,k,\off:& integer\\
    c: & character ~~~~~~&&
    \end{array} }
    \end{array}$
    \caption{Syntax of the rule description language $\langc$.   \label{fig:dslsyntax}}
\end{figure}
\begin{figure*}[!htpb]
    \[
    \arraycolsep=3pt\def\arraystretch{1.8}
    \begin{array}{rcl}
    \semantics{\fixRule{\origCmd}{\errMsg}{\fixCmd}} (\bar{s}_{cmd},\bar{s}_{err}) & = & \begin{cases}
    \semantics{\fixCmd}_\state &  \sigma = \unify(\origCmd,\bar{s}_{cmd}) \cup \unify(\errMsg,\bar{s}_{err}) \wedge \sigma \neq \bot \\
     \bot & \mbox {otherwise}
\end{cases}\\
\unify([\e_1,\cdots, \e_l],[s_1,\cdots, s_l]) & = &
\bigcup_{1 \leq j \leq l } \matchexpr(\e_j,s_j), \\
\matchexpr(\constStr(s_1), s_2) & = & \begin{cases}
    \state_0 &  \mbox{if } s_1=s_2\mbox{ and } \state_0 \mbox{ is the always undefined function}\\
     \bot & \mbox {otherwise}
\end{cases}\\
\matchexpr(\varmatch(i, \pref, \suf), s) & = & \begin{cases} \state & \mbox{ if } \exists \, \delta \mbox{ s.t. } s = \pref \delta \suf \mbox{ and } \state(i) = s \mbox{ and } \state \mbox{ is undefined on every }j\neq i \\
   \bot & \mbox {otherwise}
\end{cases}\\
\semanticsfix{[\g_1,\cdots, \g_n]}_\state & = & [\semanticsfix{\g_1}_\state,\cdots,\semanticsfix{\g_n}_\state]\\
    \semanticsfun{\fconstStr(s)}_\state & = & s \\
    \semanticsfun{\substrAppend(\pl,\pr,\pref,\suf,\varexpr(i))}_\state & = &
    \pref \cdot m \cdot \suf \qquad\mbox{if } (i,s)\in \state \mbox{ and } m=\substr(s,\semanticspos{p_L}_{s,L},  \semanticspos{p_R}_{s,R}) \\
    \semanticspos{\constPos(k)}_{s,d}  =  \begin{cases}
    k & k>0\\
    \length{s} + k & k<0\\
    0 & k=0 \wedge d=L\\
    \length{s} & k=0 \wedge d=R\\
    \end{cases}&&
    \semanticspos{\symPos(c,k,\off)}_{s,d}  =
    \begin{cases}
        \ind[k-1]+\off & k>0 \mbox{ and }  \ind=\indices(s,c) \\
                     & \mbox{and } \length{\ind}\geq k\\
        \ind[l-k]+\off & k<0 \mbox{ and }  \ind=\indices(s,c)\\
                     & \mbox{and } \length{\ind} + k\geq 0 \wedge l=\length{\ind}\\
        \bot & \mbox {otherwise}\\
    \end{cases}\\
\multicolumn{3}{c}{
    \indices(s,c) = [i_0,\ldots, i_j]  \qquad \mbox{ where }\forall ~0 \leq l < j: ~ i_l < i_{l+1}, s[i_l]=c \mbox{ and } \forall j: s[j]=c \rightarrow \exists ~0 \leq l < j: i_l=j
    }
    \end{array}
    \vspace{-3mm}
    \]
    \caption{The semantics of the command repair DSL $\langc$.
    \label{fig:dslsemantics}}
\end{figure*}
\paragraph{Fix expressions}
A fix expression $\g$ is either of the form $\fconstStr(s)$
denoting the constant output string $s$,
or of the form \linebreak
$\substrAppend(\pl,\pr,\pref,\suf,\varexpr(i))$ denoting
a function that is applied to the string $s_i$ matched by the variable $\varexpr(i)$.
This function outputs the string $\pref\cdot m\cdot \suf$, where $\cdot$ denotes the string concatenation operator and $m=\substr(s,j_L,j_R)$ 
where
$j_L$ and $j_R$ are the indices resulting from respectively evaluating the position expressions
$p_L$ and $p_R$ on the string $s$.
Here, given a string $s=a_0\ldots a_n$, the expression $\substr(s,j_L,j_R)$ denotes the string
$a_{j_L}\ldots a_{j_R-1}$ if $j_L,j_R\leq n+1$ and the undefined value $\bot$ otherwise. Notice that, unlike previous VSA-based languages~\cite{cacm12}, \langc does not allow binary recursive concatenation;
this is one of the key features that enables polynomial time synthesis.

\paragraph{Positions expressions}
A position expression $p$ can be one of the following types of expressions.
\begin{itemize}
\item A constant position expression $\constPos(k)$, which denotes the index $k$ if $k$ is positive
			and   the index $\length{s}-k$ if $k$ is negative.
			If $k=0$, this expression evaluates to $0$ when evaluated for $p_L$ (i.e., the starting index of the substring)
			and to $\length{s}$ when evaluated for $p_R$, where $\length{s}$ denotes the length of the string $s$.
			For example, in the function $\substrAppend(\constPos(0), \constPos(0),\varepsilon,\varepsilon,\varexpr(1))$,
			where $\state(1)=$~\texttt{File}
			the first constant position evaluates to the index $0$, while the second constant position
			evaluates to the index $\length{\texttt{File}}=4$.
\item A symbolic position expression $\symPos(c,k,\off)$, which denotes
			the result of applying an offset $\off$ to the index of the $k$-th occurrence of the character $c$
			in $s$ if $k$ is positive, and
			the result of applying an offset $\off$ to the index of the $k$-th to last occurrence of the character $c$
			in $s$ if $k$ is negative.
			For example, given the string \texttt{www.google.com},
			the expression $\symPos(\t{.},1,-2)$ denotes the index $2$ (two positions before the first dot),
			while
			the expression $\symPos(\t{.},-1,2)$ denotes the index $12$ (two positions after the last dot).
			This operator is novel and can express operations that are not supported by
			previous VSA-based work. In particular, FlashFill~\cite{cacm12} only allows 
			the extraction of the exact position of a character and not positions in its proximity. Despite this additional capability,
			\langc programs can be synthesized in polynomial time.
\end{itemize}

\paragraph{Comparison with FlashFill DSL}
At the top-level, $\langc$ consists of match expressions over original command and error message, which perform pattern-matching and unification of variables with strings. This form of matching and unification is not expressible in FlashFill, so we cannot use it to learn the fix rules directly. However, we can use FlashFill as a subroutine to learn string transformations corresponding to expressions similar to $\substrAppend$ expressions in \thefaq. However, the FlashFill DSL has two major limitations: 1) No support for offsets from regular expression matches in computing position expressions (in contrast to $\langc$'s $\symPos(c,k,\off)$ operator), and 2) A finite hard-coded token set for regular expressions (e.g. no support for constant character tokens).  Moreover, as described in \autoref{subsec:synthalg} and \autoref{sec:related}, our $\substrAppend$ operator yields a synthesis algorithm that operates in polynomial time in the number of examples, which enables the algorithm to scale to a large number of examples. Because of the recursive binary concatenate expressions in FlashFill, the DAG intersection based synthesis algorithm is exponential in the number of examples.
%
%
%

%% file: synthesisalgorithm.tex
\section{Synthesizing rules from examples}
\label{sec:synthalgo}

\begin{figure*}[t]
\begin{subfigure}[t]{0.49\textwidth}
{
\begin{tabularx}{\textwidth}{@{\hspace{.5em}}>{\bfseries}l@{\hspace{.5em}}X@{\hspace{.5em}}}
cmd1:  & java Run.java\\
err1: & Could not find or load main class Run.java\\
fix1: & java Run\\
\end{tabularx}
}
\caption{First example. \label{fig:rulesynth:sub:example1}}
\end{subfigure}~
\begin{subfigure}[t]{0.49\textwidth}
{
\begin{tabularx}{\textwidth}{@{\hspace{.5em}}>{\bfseries}l@{\hspace{.5em}}X@{\hspace{.5em}}}
cmd2:  & java Meta.java \\
err2: & Could not find or load main class Meta.java \\
fix2: & java Meta\\
\end{tabularx}
}
\caption{Second example. \label{fig:rulesynth:sub:example2}}
\end{subfigure}\\\\\\
\begin{subfigure}[t]{0.49\textwidth}
{
\footnotesize
\begin{tabularx}{\textwidth}{l}
\tmatch~ [\strtt{java}, \constStr(Run.java)]\\
\tand~ [\strtt{Could}, \strtt{not}, \strtt{find}, \strtt{or}, \strtt{load},\\
\qquad~ \strtt{main}, \strtt{class}, \strtt{Run.java}]\\
\tsapply~
[\fconstStr(java), ~\fconstStr(Run)]\\
\\
\\
\end{tabularx}
}
\caption{Symbolic rule representation synthesized after first example. \label{fig:rulesynth:sub:rule1}}
\end{subfigure}~
\begin{subfigure}[t]{0.49\textwidth}
{
\footnotesize
\begin{tabularx}{\textwidth}{l}
\tmatch~ [\strtt{java}, $\varmatch(1,\, \varepsilon,\, \texttt{.java})$]\\
\tand~ [\strtt{Could}, \strtt{not}, \strtt{find}, \strtt{or}, \strtt{load},\\
\qquad~ \strtt{main}, \strtt{class}, $\varmatch(2,\, \varepsilon,\, \texttt{.java})$]\\
\tsapply~
[\fconstStr(java),~$\left \{\hspace{-1mm}
  \begin{tabular}{l}
  $\substrAppend(\constPos(0),\constPos(-5), \varepsilon, \varepsilon,\varexpr(1))])$\\
  $\substrAppend(\constPos(0),\symPos(\texttt{.},1,0), \varepsilon, \varepsilon,\varexpr(1))])$\\
  $\substrAppend(\constPos(0),\symPos(\texttt{.},-1,0), \varepsilon, \varepsilon,\varexpr(1))])$
  \end{tabular}\hspace{-1mm}
\right \}$]\\
\end{tabularx}
}
\caption{Symbolic rule representation synthesized after both examples. \label{fig:rulesynth:sub:rule2}}
\end{subfigure}
\caption{
Two input examples $e_1$ and $e_2$ and
symbolic rules synthesized after processing $e_1$ and $e_2$. \label{fig:rulesynth}
}
\end{figure*}

In this section we first describe our algorithm for synthesizing a single
  $\langc$ rule from a set of examples of concrete command fixes. We then describe a multi-stage partitioning algorithm for learning multiple $\langc$ programs from a large undifferentiated set of command repair examples.

The algorithm for learning a single $\langc$ rule
is described in \figref{synthalg}; it
takes as input a list of examples $E=[e_1, \ldots, e_n]$ where each example $e_i$ is a
triple of the
form $(\scmd, \serr, \sfix)$ and outputs a symbolically represented set of $\langc$ rules $R$ consistent with $E$---i.e., for every rule $r\in R$ and example $e_i=(\scmd, \serr, \sfix)$, the rule $r$ outputs $\sfix$ on the input $(\scmd, \serr)$.

The algorithm processes one input example at a time, and after processing the first $i$ examples
$E_i=[e_1,\ldots,e_i]$
the algorithm has computed a set of rules $R_i$ consistent with $E_i$.
At the end, the algorithm outputs one of the rules in $R_n$.
We use $\bot$ to denote the undefined result. If at any point
our algorithm returns $\bot$ it means that there is no $\langc$ rule
that is consistent with the given set of examples.

\subsection{Symbolic representation of multiple rules}
Since there can be exponentially many rules consistent with the input examples,
we adopt a symbolic representation of the set $R$ that is guaranteed to always
have polynomial size.
Our synthesis algorithm takes as input a list of examples $E$
and produces
as output a symbolic rule of the form
$\sfixRule{cmd}{err}{\symfix}$, where
$cmd$ and $err$ are tuples of expressions that can consist of either constants or variables,
and $\symfix=[\g_1,\ldots, \g_m]$ is a list of expressions that symbolically represents a set of outputs  
that is consistent
with the examples $E$.
Formally, each $\g_i$ in $\symfix$ is either 
a constant expression
$\fconstStr(s)$ for some $s$, or a set of substring
expressions
$\{su_1,\ldots, su_k\}$,
where each $su_i$  is of the form
$\substrAppend(\pl,\pr,\pref,\suf,\varexpr(j))$.
Intuitively, if we replace each set with one of the fix expressions it contains, we obtain
a $\langc$ rule.
If each $\g_i$ contains $k$ elements, this symbolic representation
models $k^n$ programs using an expression of size $kn$.

\subsection{Lazy rule representation}
The core element of our algorithm is a lazy representation of the rules
that represents
match and fix expressions as constants for as long as possible---i.e.,
until a new example shows that some parts of the rule cannot be constants. This helps reduce the number of variable expressions, which in turn reduces  the number of substring expressions to be considered. We first illustrate the idea with a concrete example.
Let's say that we are given the two examples shown in~\figref{rulesynth:sub:example1} and \ref{fig:rulesynth:sub:example2}.
After processing the first example, our algorithm synthesizes the $\langc$ rule in~\figref{rulesynth:sub:rule1}
in which every match expression and every fix expression is a constant.
However, since we have only seen one example, we do not yet know whether some expression appearing in
the match should actually be a variable match expression or whether some element
in the fix expression should actually be a function of some variable.
The main idea is that any of these possibilities can still be ``recovered''
when a new example shows that indeed a variable is needed.
Using this idea, we maintain
each expression as a constant until a new example shows that some expression
cannot actually be a constant.

This is exactly what happens when processing the input example in~\figref{rulesynth:sub:example2}.
At this point in order to find a rule that is consistent with both examples we need to introduce
a variable match as the second expression of the command match,
and some function application as the second element of the fix.
To do so, our algorithm applies the following operations to the previously computed rule.
\begin{enumerate}
\item All match expressions that cannot be constants are ``promoted'' to variable match expressions (making sure that
		all variable names are unique), which match on the longest shared prefix and suffix of all previously seen values at that position.
The following table illustrates the idea for the case in which we try to unify the command
\textbf{cmd2} in \figref{rulesynth:sub:example2}
with the matching part of the already computed rule in
\figref{rulesynth:sub:rule1}.

\begin{table}[!h]
\centering
\begin{tabular}{rcc}
rule:  & \strtt{java} & \strtt{Run.java}\\
new-ex: & java & Meta.java \\
\hline
new-rule: & \strtt{java} & $\varmatch(1, \varepsilon, \texttt{.java})$
\end{tabular}
\end{table}

\item All the fix expressions that cannot be constants are ``promoted'' to
		  \substrAppend~ expressions that are consistent with the current examples and
		  are allowed to use the variables appearing in the match expressions.		
\end{enumerate}

The second rule in~\figref{rulesynth}(d) reflects this update.
The figure also shows how multiple $\substrAppend$ expressions are represented symbolically as a set.
We describe all of these components in detail in the next section.

\subsection{Synthesis algorithm}
\label{subsec:synthalg}
\begin{figure*}[t]
\begin{subfigure}[t]{.49\textwidth}
{
\begin{algorithmic}
\fcomm{Rules consistent with input examples}
\Function{SynthRules}{$[e_0,\ldots,e_n]$}
	\State $r \gets \Call{ConstRule}{e_0}$
    \For{$1\leq i  \leq n$} \Comment{refine on each example $e_i$}
    	\State $r \gets \Call{RefineRule}{r,[e_0,\ldots, e_{i-1}],e_i}$
    \EndFor
    \State \Return $r$
\EndFunction
\fcomm{Constant rule consistent with one example}
\Function{ConstRule}{$[s^c_1,\ldots, s^c_n],[s^e_1,\ldots, s^e_m],[s^f_1,\ldots, s^f_l]$}
	\State $cst_{cmd} \gets [\constStr(s^c_1),\ldots,\constStr(s^c_n)]$
    \State $cst_{err} \gets [\constStr(s^e_1),\ldots,\constStr(s^e_m)]$
    \State $cst_{fix} \gets [\fconstStr(s^f_1),\ldots,\fconstStr(s^f_l)]$
    \State \Return ($\sfixRule{cst_{cmd}}{cst_{err}}{cst_{fix}}$)
\EndFunction
\fcomm{Refines a rule to make it consistent with one more example}
\Function{RefineRule}{$r,E,(\scmd,\serr,\sfix)$}
	\State $r \equiv (\sfixRule{cmd}{err}{\symfix})$
	\State $(cmd',V_c) \gets \Call{FindVariables}{\scmd,cmd,0}$
	\State $(err',V_e) \gets \Call{FindVariables}{\serr,err,\length{\scmd}}$		
	\State $V\gets V_c\cup V_e$
	\State $E' \gets (\scmd,\serr,\sfix):: E$
    \State $\symfix' \gets \Call{SynthFix}{\sfix,\symfix,E',V}$
    \State \Return ($\sfixRule{cmd'}{err'}{\symfix'}$)
\EndFunction

\fcomm{Finds variables necessary to match example}
\Function{FindVariables}{$[s_1,\ldots, s_n],[t_1,\ldots, t_m],o$}
	\If {$n\neq m$}\Comment{Input length same as match length?}
		\State \Return $\bot$
	\EndIf
	\State $(m,V) \gets ([], \emptyset)$
	\For{$1\leq i  \leq n$}
 		   \Case{$t_i=\constStr(s)\wedge s= s_i$}
		      	\State $(m,V) \gets (m @ \constStr(s_i), V)$
 		   \EndCase
 		   \Case{$t_i=\constStr(s)\wedge s\neq s_i$}
 		        \State $pref \gets \Call{LongestCommonPrefix}{s, s_i}$
                        \State $suf \gets \Call{LongestCommonSuffix}{s, s_i}$
 		        \State $newId \gets i+o$ \Comment{Create new variable}
                        \State $m \gets  m@[\varmatch(newId,\, pref,\, suf)]$
 		        \State $V \gets V\cup \{newId\}$
 		   \EndCase
 		   \Case{$t_i=\varmatch(j, l, r)$}
                        \State $pref \gets \Call{LongestCommonPrefix}{s, l}$
                        \State $suf \gets \Call{LongestCommonSuffix}{s, r}$
                        \State $m \gets m@[\varmatch(j,\, pref,\, suf)]$
                        \State $V \gets V\cup \{j\}$
 		   \EndCase
         \EndFor
    \State \Return $(m,V)$
\EndFunction    
    
\end{algorithmic}
}
\end{subfigure}
~~~~
\begin{subfigure}[t]{.49\textwidth}
{
\begin{algorithmic}
\fcomm{Outputs the fixes consistent with
				all the examples $E$ and such that $\substrAppend$ expressions
				can depend on any variable in $V$.
				The fix component of the latest example and the fixes computed on the previous examples are also
				passed as input}
\Function{SynthFix}{$[s_1,\ldots, s_n],[t_1,\ldots, t_n],e::E,V$}
	\If {$n\neq m$}
		\State \Return $\bot$
	\EndIf
	\State $f \gets []$
	\For{$1\leq i  \leq n$}
 		   \If{$t_i=\fconstStr(s)\wedge s= s_i$}
		      	\State $f \gets f@[\fconstStr(s_i)]$
 		   \Else
 		   		\State $f \gets f@[\Call{SynthSubstrings}{e::E,V,i}]$
 		   \EndIf
    \EndFor
    \Return $f$
  \EndFunction

\fcomm{Outputs all $\substrAppend$ expressions consistent with
				all the examples that can appear
				at position $i$ in the fix expression}
\Function{SynthSubstrings}{$e::E,V,i$}
	\State $F \gets \Call{AllSubstrings}{e,V,i}$
	\ForAll{$(\scmd,\serr,\sfix)\in E$}
			\State $F'\gets \emptyset$
			\ForAll{$fun\in F$}
				\State $\textbf{let }\substrAppend(\pl,\pr,\pref,\suf,\varexpr(j))=fun$
				\State $o\gets \length{\scmd}$
				\If {$j<o$} \Comment{The variable is in $\scmd$}
					\If {$\Call{eval}{fun,\scmd[j]=\sfix}$}
						\State $F'\gets fun::F'$
					\EndIf								
				\Else	 \Comment{The variable is in $\serr$}
					\If {$\Call{eval}{fun,\serr[j-o]=\sfix}$}
						\State $F'\gets fun::F'$
					\EndIf		
				\EndIf	
			\EndFor
			\State $F \gets F'$
	\EndFor
	\State \Return $F$
\EndFunction

\fcomm{Outputs all $\substrAppend$ expressions consistent with
				one example that can appear
				at position $i$ in the fix expression}
\Function{AllSubstrings}{$(\scmd,\serr,\sfix),V,i$}
    \State \Return \mbox{all }$\substrAppend(\pl,\pr,\pref,\suf,\varexpr(j))$ \mbox{ that }
    \State ~~~~~\mbox{when evaluated on $\scmd$ and $\serr$ output $\sfix[i]$ }
    \State ~~~~~\mbox{and such that }$j\in V$.
\EndFunction
    
\end{algorithmic}
}
\end{subfigure}
\vspace{-2mm}
\caption{Algorithm for synthesizing $\langc$ rules from concrete examples.\label{fig:synthalg}}
\end{figure*}

Given a list of input examples, the function $\textsc{SynthRules}$ uses the first example and the function
\textsc{ConstRule}
to generate the symbolic rule composed only of constant operators.
It then iteratively refines the rule on the remaining examples as shown in Figure~\ref{fig:synthalg}.
This second operation is done by the function \textsc{RefineRule}
which takes as input a symbolic rule $r$, one new example
$(\scmd,\serr,\sfix)$, and the list of examples $E$
on which every concrete rule represented by $r$ behaves correctly.
\textsc{RefineRule} executes two main steps using the following functions.

\paragraph{\textnormal{\textsc{FindVariables}}} tries to unify the inputs $\scmd$ and $\serr$
with the corresponding match expressions $cmd$ and $err$ in the symbolic rule $r$ and
generates new variable match expressions if necessary---i.e., when $r$ contains
a matching expression $\constStr(s)$ but the corresponding component in the example
is a string different from $s$.
In this case, a $\varmatch(i,l,r)$ expression is generated such that
$i$ is a new variable name, and
$l$ and $r$ are the longest prefix and suffix shared by $s$, respectively. 

When \textsc{FindVariables} is presented with a new $\scmd$ or $\serr$ after a constant match expression has been `promoted' to a \linebreak $\varmatch(i,l,r)$ expression, the prefix and suffix are updated accordingly. \textsc{FindVariables} determines the longest prefix $r'$ and suffix $l'$ of $l$ and $r$, respectively, that is consistent with the appropriate component of the new example, and generates $\varmatch(i,l',r')$.

\paragraph{\textnormal{\textsc{SynthFix}}} uses the variables computed in the previous step
to synthesize all possible fix expressions that are consistent with
			the list of examples $\{(\scmd,\serr,\sfix)\}::E$.

In order to simplify variable naming and guarantee unique names, each variable has the index of the corresponding
element in the input---i.e., $\varexpr(i)$ denotes the $i$-th string in the list $\scmd @ \serr$ obtained by
concatenating the command and error input lists.

\paragraph{Lazy pattern matching}
The function \textsc{FindVariables}, given a rule $r$ and a new example $e$, iterates over the input components
of the new example $e$ and outputs the set of variables necessary to match this new example
with respect to the previously computed symbolic rule $r$.
The function \textsc{SynthFix}, given a rule $r$ and a list of examples $E$,
individually synthesizes all the components $f_i$
of the symbolic output fix expression that are consistent with $E$.
\textsc{SynthFix} is incremental in the sense that it tries to minimally change the original fix expression of $r$:
\begin{itemize}
\item if the $i$-th component $t_i$ of the fix expression of $r$ is a constant string consistent with the new example,
			then
			it is left unchanged;
\item in any other case the output has to be a substring operation, and the
			function \textsc{SynthSubstrings} is used to compute all the possible
			$\substrAppend$
			expressions that are consistent with the set of examples $E$.
\end{itemize}
Given a list of examples $e::E$, the function \textsc{SynthSubstrings}
first synthesizes all the $\substrAppend$ expressions
that are consistent with $e$ using the function $\textsc{AllSubstrings}$ and
then runs each synthesized expression on examples in $E$ to remove
 the inconsistent ones.

\paragraph{\textup{ALLSUBSTRINGS}}
\figref{synthalg} omits the formal definition of the function $\textsc{AllSubstrings}$ due to  space limitations,
but we describe its main components.
Given  an example \linebreak $e=(\scmd,\serr,\sfix)$, a set of variable names $V$, and the index $i$ corresponding to the element of the output
sequence we are trying to synthesize,
$\textsc{AllSubstrings}$ computes the set of all substring expressions
of the form
$fun=\substrAppend(\pl,\pr,\pref,\suf,\varexpr(j))$ that are consistent with $e$ such that the result of applying $fun$ to the $j$-th string in $\scmd @\serr$ is the
$i$-th string in $\sfix$.
Let's assume that $\length{\scmd}+\length{\serr}=n_I$, $\length{\sfix}=n_O$,
and $n_L$ is the length of the longest string appearing in any of the three lists in the input example.
To compute the set $\textsc{AllSubstrings}(e,V,i)$ we iterate over all variable indices
and for each variable index $j\in V$ we do the following.
\begin{enumerate}
\item Extract the string $s_j$ corresponding to the variable $\varexpr(j)$ --- $\bigO{n_I}$ iterations.
\item For each string $s$ that is a substring of both $\sfix[i]$ and $s_j$,
			compute all possible pairs of indices $k_1,k_2$ such that $\substr(s_j,k_1,k_2)=s$ --- $\bigO{n_K^2}$ 
            possible substrings and $\bigO{n_K}$ possible ways to
				 place the substring in $\sfix[i]$.
\item For each $k_1$ (resp. $k_2$) compute every position expression
			$p_1$ (resp. $p_2$)
			such that evaluating $p_1$ (resp. $p_2$)
			 on $s_j$ produces the index $k_1$ (resp. $k_2$) --- $\bigO{n_K}$ possible positions.
\item For each of these possibilities yield the expression\linebreak
			$\substrAppend(p_1,p_2,l,r,\varexpr(j))$
			where $l$ and $r$ are such that $\sfix[i] = l \cdot \substr(s_j,k_1,k_2) \cdot r$.
\end{enumerate}
\textsc{AllSubstrings} produces a set of expressions that in the worst case has size
$\bigO{n_In_K^5}$\footnote{
Note that the efficiency of this implementation is contingent on our specific choice of data structure and algorithms. A more naive solution, based on set intersection (like that of the DAG-based algorithm in FlashFill) may experience exponential blow-up in the number of examples, due to the quadratic nature of the intersection operations.
}.
If we restrict the offset component $\off$ to only range over the values $\{-1,0,1\}$ for the symbolic expressions $\symPos(c,i,\off)$, the size reduces to $\bigO{n_In_K^3}$, and the synthesis algorithm is still sound and complete for this fragment of $\langc$.

This last restriction of the language can capture all the rules we are interested in.
Notice that this analysis still holds in the extreme case in which
all input matches are variable expressions of the form $\varmatch(i, \varepsilon, \varepsilon)$.
In our experiments on real-world commands, worst-case performance is uncommon, and is induced by substring operations over heterogeneous strings which yield many possible implementations. Consider the following two examples.

\begin{table}[H]
\begin{tabularx}{\columnwidth}{|@{\hspace{.5em}}>{\bfseries}l@{\hspace{.5em}}X@{\hspace{.5em}}|@{\hspace{.5em}}>{\bfseries}l@{\hspace{.5em}}X@{\hspace{.5em}}|}
\hline
cmd1:  & aaaa aaaa & cmd2: & bbbb bbbb\\
err1:  & aaaa aaaa & err2: & bbbb bbbb\\
fix1:  & aa        & fix1: & bb\\
\hline
\end{tabularx}
\end{table}
Performing synthesis on this pair of examples yields a pattern match consisting of four $\varmatch(i,\epsilon,\epsilon)$ expressions. Due to the uniformity of the input strings, synthesis yields $48$ possible $\substrAppend$ expressions. In particular, the desired fix can be generated from any of the four strings in the supplied $s_{cmd}$ and $s_{err}$. Each of the four strings has three substrings of length $2$, any of which yields the desired output. For each such substring, there are four pairs of $\constPos$ values that supply the appropriate indices: The pair with two positive indices, the pair with two negative indices, and the two pairs consisting of one positive and one negative index.

\paragraph{Key point}
At this point we are ready to explain why all the match expressions can be kept as constants
for as long as possible.
If after processing a set of examples $E$, some expression in $cmd$ or $err$
is of the form $\constStr(s_i)$, then, for every input example,
the value of the $i$-th component is the string $s_i$.
Therefore, even if we replace this expression with a variable,
all its instantiations will have the same values.
Consequently, every function of the form  $\substrAppend(p_1,p_2,l,r,\varexpr(i))$
will produce a constant output, making it equivalent to the some
constant function $\fconstStr(s')$. 

\subsection{Partition-Based Synthesis} \label{partsynth}
\langc can learn repair rules from a set of examples corresponding to a specific incorrect use of a command. In practice, however, it may be difficult to present \langc with a collection of neatly curated sets of examples, from each of which, \langc learns a single symbolic rule. Such a process is both labor-intensive and error-prone. We instead envision large-scale learning of \langc rules from undifferentiated sets of examples submitted by many users. To facilitate this, we propose a simple multi-stage partitioning strategy.
As a consequence of \langc's structure, each symbolic rule matches on a pair of command and error strings, each of which has a fixed number of tokens. Upon a match, \langc generates a repaired command with a fixed number of tokens. Conversely, every symbolic \langc expression must be learned from a set $S_E$ of example triples of the form $(s_{cmd}, s_{err}, s_{fix})$ for which the lengths of $s_{cmd}$, $s_{err}$, and $s_{fix}$ do not vary.

Given an undifferentiated set of examples, $S_E$, we partition $S_E$ into $n$ disjoint subsets $\delta_i$ where $S_E = \bigcup_{i=0}^{n} \delta_i$. For every $\delta_i$ the property $\forall (s_{cmd},s_{err},s_{fix}) \in \delta_i, \,\,  (\scmd, \serr,\sfix) \in \delta_i .\, |s_{cmd}| = |\scmd|\, \wedge \, |s_{err}| = |\serr|\, \wedge\, |s_{fix}| = |\sfix|$. This divides $S_E$ into subsets from which it is possible to synthesize \langc rules.

After dividing $S_E$, it may still be the case that individual sets $\delta_i$ contain examples representing distinct command repair rules which share the same triple of $s_{cmd}$, $s_{err}$, and $s_{fix}$ lengths. At this point, we attempt to find the smallest set of rules that can be synthesized from the examples in $\delta_i$.

The search ranges over all partitions $P$ of $\delta_i$, where $P$ is the set $\{\delta_{i,1}, \delta_{i,2}, \ldots, \delta_{i,m}\}$ such that $\bigcup_{j=1}^{m} \delta_{i,j} = \delta_i$ and $\forall j \neq k, \delta_{i,j} \cap \delta_{i,k} = \varnothing$. We enumerate the partitions in ascending order of size $m$, starting with $P = \{\delta_i\}$, and ending with the partition consisting entirely of singleton sets. Given a partition $P$, we attempt to synthesize a symbolic \langc rule for each set in $P$, stopping when we have generated a rule for each $\delta_{i,j}$. As we will show in section~\ref{sec:evaluation}, this expensive procedure is only feasible when using our novel contribution of lazy VSA.

\paragraph{Example partitioning}

\begin{figure}[!t]
\begin{tabularx}{\columnwidth}{r@{\hspace{.5em}}>{\bfseries}l@{\hspace{.5em}}X@{\hspace{.5em}}}

&cmd1: & java Run.java \\
$e_1$ & err1: & Could not find or load main class Run.java\\
&fix1: & java Run\\
\hline
&cmd2: & java Test.java \\
$e_2$ & err2: & Could not find or load main class Test.java\\
&fix2: & java Test\\
\hline
&cmd3: & composer pkg\\
$e_3$ & err3: & did you mean one of these? pkg1 pkg2\\
&fix3: & composer pkg1\\
\hline
& cmd4: & composer hptt\\
$e_4$ & err4: & did you mean one of these? http html\\
& fix4: & composer http\\
\hline
& cmd5: & mv photo.jpg Mary/summer12.jpg \\
$e_5$ & err5: & can't rename `photo.jpg': No such file or directory\\
& fix5: & mkdir Mary \texttt{\&\&} mv photo.jpg Mary/summer12.jpg\\
\hline
& cmd6: & mv dec31.jpg Bob/family.jpg \\
$e_6$ & err6: & can't rename `dec31.jpg': No such file or directory\\
& fix6: & mkdir Bob \texttt{\&\&} mv dec31.jpg Bob/family.jpg\\
\end{tabularx}
\caption{Examples requiring more than one \langc rule. \label{fig:multiruleex}}
\end{figure}
Consider the example set shown in Figure~\ref{fig:multiruleex}.
The first two examples in this set correspond to the repair in Section~\ref{sec:motiv:adding} 
while the last two correspond to the repair in Section~\ref{complexsubs}. 
The third and fourth example correspond to the rule that outputs \texttt{composer}, followed by the token at index $8$ in the input 
(i.e., the first token suggested by the error message).

We first group the examples based on the length of their components. 
For the first four examples, we have $|s_{cmd}| = 2$, $|s_{err}| = 8$, $|s_{fix}| = 2$
and for the two remaining examples, 
$|s_{cmd}| = 3$, $|s_{err}| = 8$, $|s_{fix}| = 6$. Thus, we obtain
two groups $S_1 = \{e_1, e_2, e_3, e_4\}$ and $S_2 = \{e_5, e_6\}$.

While there exists a \langc program that is consistent with the examples in the set $S_2$,
no \langc program can describe a transformation that is consistent with all the examples in $S_1$.
We therefore proceed by iteratively partitioning the set $S_i$, attempting to find a partition $P_i$ for which we can synthesize a rule for every $S_{i,j} \in P_i$. For $S_2$, we can clearly do so for the initial partition $P_2 = \{ S_2 \}$, yielding the rule from Section~\ref{sec:motiv:extractingsubs}. For $S_1$, the first partition for which we can generate a \langc rule for every element is \linebreak $P_1 = \{ S_{1,1} = \{e_1, e_2\}, \, S_{1,2} = \{e_3, e_4\} \}$. $S_1$ yields the rule from Section~\ref{sec:motiv:adding}, and $S_2$ yields the simple substitution rule described previously.

\subsection{More Succinct Representation}
The version of symbolic rule we presented is already able to store exponentially many
concrete $\langc$ rules in polynomial space.
In this section, we discuss further improvements that can make the representation
more succinct.

\paragraph{Avoid redundancy}
In the set of fix expressions enumerated by the function \textsc{AllSubstrings}, the last three components of the
the expression $\substrAppend(p_1,p_2,l,r,\varexpr(j))$ are often repeated many
times.
Looking at \figref{rulesynth:sub:rule2} we can see how
all the synthesized functions have $l=r=\varepsilon$ and are applied to
the variable $\varexpr(1)$.
We define a data structure for representing sets of fix expressions that avoids these repetitions.
A set of fix expressions is represented symbolically 
using a partial function
$$
d: \mathbb{N} \mapsto (\Sigma^* \times \Sigma^*) \mapsto Set(P\times P)
$$
where $P$ is the set of all position expressions.
Formally, given a variable index $i$ and
two strings $l$ and $r$,
the set $d(i,l,r)$ symbolically represents
the set of
fix expressions \linebreak $\{\substrAppend(p_1,p_2,l,r,\varexpr(i))\mid (p_1,p_2)\in d(i,l,r)\}$.
The function $d$ can be efficiently implemented and avoids redundancy.
Considering again the example rule in
\figref{rulesynth:sub:rule2},
all the fix expressions in the second component of the output can be succinctly represented by the function
$d$ that is only defined
on the input $(1,\varepsilon,\varepsilon)$ and such that
$$
\begin{array}{rll}
d(1,\varepsilon,\varepsilon) = \{&
	(\constPos(0),\constPos(-5)),\\
	&(\constPos(0),\symPos(\texttt{.},1,0)),\\
	&(\constPos(0),\symPos(\texttt{.},-1,0))&\}.\\
\end{array}	
$$

\paragraph{Avoid example representation}
Each synthesized rule is
in some sense coupled to the set of examples used to synthesize it.
We present a data structure that only keeps track of the important ``parts'' of the input examples
and therefore allows us to discard each example after it has been processed.

We modify the symbolic rule representation as follows.
Given a set of examples $(\scmd^1,\serr^1,\sfix^1)\ldots (\scmd^n,\serr^n,\sfix^n)$,
\begin{itemize}
\item every variable $\varmatch(i,l,r)$ in the match component becomes
a pair $(\varmatch(i, l, r),[b_1,\ldots,b_n])$ where the second component is the list of strings that binds to
$\varexpr(i)$ in the input components of the examples---i.e. $b_j=(\scmd^j\serr^j)[i]$;
\item every set of fix expressions represented by a function $d_i$
and corresponding to the $i$-th
component of the fix expression becomes a pair
$(d_i, [b_1,\ldots,b_n])$ where the second component is the list of strings that appear
in position  $i$ in the output components of the examples---i.e. $b_j=\sfix^j[i]$;
\end{itemize}
Using this data structure we do not need to store examples as we can always re-infer
them from the symbolic rule representation.

\subsection{Concrete outputs}
Taking into account the updated data structures,
the algorithm $\textsc{SynthRules}$ returns
a symbolic rule $r$
of the form \linebreak $\sfixRule{cmd}{err}{\symfix}$
where $cmd=[c_1,\ldots,c_n]$ and $err=[e_1,\ldots,e_m]$ are lists of expressions of the form
$\constStr(s)$ or $(\varexpr(i),B)$, while $\symfix=[f_1,\ldots,f_l]$ is a list
of expressions of the form
$\fconstStr(s)$ or $(d,B)$ where $d$ is the
data structure for representing multiple fix
expressions.
The set of concrete $\langc$ rules induced by this symbolic representation
is the following.
$$
\begin{array}{l}
\conc(\sfixRule{cmd}{err}{fix}) =\\
\qquad\qquad \{\fixRule{cmd}{err}{f}\mid  f\in \conc(fix)\}\\
\conc([f_1,\ldots, f_l]) =\{[f_1',\ldots,f_l']\mid f_i'\in \conc(f_i)\}\\
\conc(\fconstStr(s)) =\{\fconstStr(s)\}\\
\conc(d,B) =\{\substrAppend(p_1,p_2,l,r,\varexpr(i)) \mid\\
\qquad \qquad\qquad \qquad  \qquad \qquad  \exists i,l,r.(p_1,p_2)\in d(i,l,r)\}.
\end{array}
$$

\section{Formal properties}
\label{sec:algoproperties}

We study the formal properties of the synthesis algorithm and of the language
$\langc$. 
These specific properties describe the behavior of the synthesis algorithm in the absence of the partitioning strategy described in Section~\ref{partsynth}.
\paragraph{Properties of the synthesis algorithm}
First, our synthesis algorithm is invariant with respect to the order
in which the training examples are presented. Thus, the properties of a symbolic rule generated by $\textsc{SynthRules}$, can be discussed solely in terms of the \textit{set} of examples provided to $\textsc{SynthRules}$.
\begin{theorem}[Order invariance]
Given a list of examples $E$, for every permutation of examples $E'$ of $E$, we have \linebreak
$\conc(\textsc{SynthRules}(E))=\conc(\textsc{SynthRules}(E'))$.
\end{theorem}
\begin{proofsketch}
  Consider a list of examples $E$. If two examples differ at the $i$th position in their respective commands, the $i$th expression in $cmd$ will be promoted to a $\varmatch$, regardless of the order in which they are presented to $\textsc{SynthRules}$. Moreover, if all $|E|$ strings at the $i$th position in the commands share a prefix or suffix, reordering $E$ does not change this fact. Thus, the discovered $\varmatch$ expressions will not vary based on order. The same holds for $err$ and the ``promotion'' of constants in $\textsc{SynthFix}$.
  
Moreover, $\textsc{SynthSubstring}$ starts from scratch at each iteration of the loop in $\textsc{SynthRules}$ and $fix$ depends only on the output of $\textsc{SynthFix}$ in the final iteration. Since the variable set does not vary based on the ordering of $E$, the final invocation of $\textsc{SynthFix}$ does not depend on the ordering of $E$.\qed
\end{proofsketch}
Second, the synthesis algorithm produces only rules that are consistent
with the input examples. If we select an arbitrary concrete rule $r$ from the set specified by a symbolic rule generated by $\textsc{SynthRules}$, and run it on the command and error of any of the examples provided to $\textsc{SynthRules}$ for the synthesis of $r$, we will obtain the fix originally provided in that example.
\begin{theorem}[Soundness]
Given a list of examples $E$, for every rule $r\in\conc(\textsc{SynthRules}(E))$
and for every example \linebreak $(\scmd,\serr,\sfix)\in E$,
$\semantics{r}(\scmd,\serr)=\sfix$.
\end{theorem}
\begin{proofsketch}
  The repeated applications of $\textsc{FindVariables}$ will promote any $\constStr(s)$ expression if a new example does not match on $s$. Moreover, when refining a $\varmatch$, $\textsc{FindVariables}$ chooses the longest prefix and suffix consistent with all examples seen so far. Thus, $cmd$ and $err$ will correctly match on all examples. The soundness of the resulting $fix$ derives from the fact that at each iteration of the loop in $\textsc{SynthRules}$, the invocation of $\textsc{SynthSubstring}$ in $\textsc{SynthFix}$ takes into account all examples seen in previous iterations of the loop. Moreover, each invocation begins with the set of all possible $\substrAppend$ expressions, and prunes those inconsistent with any example seen so far.\qed
\end{proofsketch}
Since parts of the match expressions are ``promoted'' to variables
only when the input examples show that this is required,
our synthesis algorithm does not explicitly keep track of all the possible
rules that can be consistent with the examples. 
Our completeness result reflects this idea.
\begin{theorem}[Completeness]
Given a set of examples $E$, for every $\langc$ rule $r$  that is consistent
with $E$,
either \linebreak $r\in\conc(\textsc{SynthRules}(E))$ or
there exists an example $e$ such that
$r\in\conc(\textsc{SynthRules}(e::E))$.
\end{theorem}
\begin{proofsketch}
Concretely, a particular rule $r$ that is consistent with $E$ might not appear in $R=\textsc{SynthRules}(E)$.
However, this can only happen because the match expression of $r$ has more variables than the match of any rule in $R$.
This can be fixed by providing an example that forces
the algorithm to promote to variables all the required match expressions. \qed
\end{proofsketch}

\paragraph{Properties of the language $\langc$}

We define the size of an input example $e=(\scmd,\serr,\sfix)$ and
the size of each rule $r$ consistent with $e$
as the sum of its lengths $\size(e)=\size(r)=|\scmd|+|\serr|+|\sfix|$.
For each set of examples, there can be exponentially
many rules consistent with it.
\begin{theorem}[Number of consistent rules]
Given a set of examples $E$, each of size $k$,
the set
$\conc(\textsc{SynthRules}(E))$ contains
$2^{\bigO{poly(k)}}$ rules.
\end{theorem}
\begin{proofsketch}
  As we showed in Section~\ref{subsec:synthalg}, for each position in the output of a \langc rule there are potentially polynomially many $\substrAppend$ expressions consistent with the provided examples. For the $i$th position in $fix$, we are free to choose any of the possible $\substrAppend$s, independently of our choice at other positions. Thus, the number of possible rules is potentially exponential in $|\sfix|$.\qed
  \end{proofsketch}

 Despite the exponential number of rules represented, our data structure allows the \textsc{SynthRules} to encode these rules using only polynomial size.

Next, there exists an active learning algorithm for learning $\langc$ rules that requires only a polynomial number of examples---i.e.,
queries to the user.
\begin{theorem}[Complexity of active learning]
If there exists a target rule $r$ of size $k$,
there exists an active learning algorithm that
will learn $r$ by asking $\bigO{poly(k)}$
queries of the form:
\emph{What should the output of $r$ be on the input $(\scmd,\sfix)$}.
\end{theorem}
\begin{proofsketch}
The algorithm first asks $k$ queries to figure out which match expressions are variables and which ones are
constants. Then, for each output component for which there exists two possible fix expressions $\substrAppend$
consistent with the examples, it asks a query that differentiates the two.
Since there are only $\bigO{kn_In_K^5}$ many expressions in the output
the algorithm will ask at most polynomially  many  queries.\qed
\end{proofsketch}

%
%

%% file: evaluation.tex
\section{Implementation and Evaluation}
We now describe the implementation details of \thefaq, as well as our experimental evaluation of \thefaq on a set of examples and test cases isolated from \thefuck and web forums.
\subsection{Implementation}
 We implemented the language $\langc$ and its synthesis algorithm in a system
called \thefaq. $\thefaq$ is implemented in F\#\footnote{The implementation will be made open-source and publicly available after the review
process.} and consists of some additional optimizations and design choices as described below.
\subsubsection{Implementation optimizations}
The function \textsc{AllSubstrings} in \figref{synthalg} synthesizes all $\substrAppend$
 functions that are consistent with the first input/output
pair $(s,t)$ of strings in the example set $E$ and then applies each of the synthesized functions
to the other elements of $E$ for filtering only consistent functions. In practice, we first compute the
longest common prefixes and suffixes of  the strings appearing in the components $\sfix$ of $E$
to avoid enumerating instances of the form $\substrAppend(\_,\_,l,r,\_)$ such that $l$ or $r$ are not prefixes
or suffixes of some output string $t$ appearing in $E$.

The other optimization is based on the following property of the \textsc{RefineRule} function:
when adding a new example to $r$, if the function \textsc{FindVariables} introduces a new set of variables $V$,
all the new instances of $\substrAppend$ that did not already appear in $r$ depend on one of the newly introduced
variables in $V$.
Based on this idea, the function \textsc{AllSubstrings} only has to compute functions of the form
  $\substrAppend(\_,\_,\_,\_,\varexpr(i))$ where $i\in V$, and can reuse the previously computed functions for the other variables by simply filtering the ones that behave correctly on the newly introduced example.

\subsubsection{Ranking}
\label{sec:ranking1}
Since there can be multiple possible expressions in \langc that are consistent with the examples, we employ a simple ranking technique to select an expression amongst them. If there are multiple $\substrAppend$ expressions that can generate the desired output string, we select the expression that uses the variable with the lowest index---i.e., the leftmost one. Similarly, the $l$ and $r$ included in $\varmatch$ expressions implicitly encode all rules matching on prefixes and suffixes of $l$ and $r$, respectively. We choose the longest $l$ and $r$ over all others.

As the example set increases in size, we envision users will likely submit diverse sets of examples, particularly in use cases with thousands of users submitting examples. As users submit examples which draw from heterogeneous collections of command parameters, $\varmatch$ prefixes and suffixes should converge to the least restrictive versions. Similarly \thefaq should discover the least restrictive set of constants for both match expressions. As these input parameters vary over the set of examples, spurious ambiguities in $\substrAppend$ should be eliminated when \thefaq is presented with specific fix examples which function as counter-examples to unnecessary substring expressions.

\subsection{Evaluation}
\label{sec:evaluation}
We now describe our experimental evaluation. The experiments were run on an Intel Core i7 2.30GHz CPU with 16 GB of RAM. We present both qualitative and quantitative analysis of the algorithm. 
We assess the expressiveness of \thefaq by attempting synthesis on a benchmark suite that includes the rules in the tool \thefuck. We then evaluate the performance of \thefaq and its scalability.

\subsubsection{Benchmark Suite}
We compiled our benchmark suite from an initial set of of 92 benchmarks, which were collected from both $\thefuck$ (76) and online help forums (16). We considered the 76 repair rules hard-coded in the \thefuck tool to assess the expressiveness of \thefaq. Since rules in \thefuck can use arbitrary Python code, it is hard to exactly compare them to the ones produced by \thefaq. We use manual testing to check that a rule $r$ generated by our tool is \emph{consistent} with a rule $r'$ in \thefuck. To do so, we manually constructed a set of examples based on the pattern-matching and textual substitutions performed by the \thefuck rules.

The other sixteen example sets were obtained from examples found during a non-exhaustive survey of command-line help forums on the web. These commands consist of various types of git, svn, and mvn commands, including committing, reverting, and deleting from repositories, as well as installing and removing packages.

The \thefaq system is able to synthesize a rule for 81 of the 92 benchmarks. The remaining 11 failing benchmarks can be divided into three broad categories: i) Hard-coded operations searching for specific strings in some context (8), ii) Complex patterns checking relationships between variable expressions (2), and iii) Error messages displaying parts of the input file's content (1). We did not provide examples for these 11 rules. We elaborate more about these categories in Section~\ref{sec:limitations}.
\paragraph{Number of examples}
We observed that it was natural to provide two to five examples per benchmark for $\thefaq$ to uniquely learn the desired fix rule. We also provided additional examples for manually testing the learned rules, yielding a set of three to six examples. Given the rule appearing in \figref{rulesynth}, for example, we used the two examples in \figref{rulesynth} and another example with the file name Employee.java. In future, we envision users to contribute different examples to the system for automatically building a large corpus of learned fix rules.

While these examples are synthetic examples reverse engineered from other sources, they are also natural examples which exercise the range of e.g. path and file names one would expect to see in a real Unix system. In the case of the repaired command in Section~\ref{complexsubs}, the natural two-example set would consist of two distinct directory names which do not share prefixes and suffixes, as well as filenames with distinct prefixes and extensions. 

\subsubsection{Qualitative evaluation}
Given a single set containing examples for all the 81 cases in which $\thefaq$ is capable of synthesizing a rule, we performed synthesis as described in Section~\ref{partsynth}. 
For each rule we retained a single example from the training set and used it to test the accuracy of each rule. 
We also report how often a given input could be repaired using more than one rule.

\paragraph{Results} For all 81 cases, \thefaq synthesized a rule consistent with the corresponding \thefuck rule or web forum answer. In some cases we had to synthesize more than one $\langc$ rule to capture the different possible
behaviors of a single rule in \thefuck. For example, one can try adding `sudo' in front of a command
for several possible errors such as ``Command not found'', ``You don't have the permission" etc.
In such cases, thanks to the partitioning algorithm, \thefaq generated a separate rule for each possible error message. 
For each case where we synthesized a rule, correctness was independent of our choice of examples. 
If the correct rule was synthesized, it was synthesized regardless of which subset of the examples provided for that rule we selected.

\paragraph{Distribution of rule sizes}
We define the size of an expression such as $\scmd$, $\serr$, and $\sfix$ as the number of strings present in it. The distribution of the size of the benchmarks in terms of the sizes of the $\scmd$, $\serr$, and $\sfix$ tuples in input-output examples is shown in \figref{indsizebenchmarks}. Note that we do not show two benchmarks in the graph with disproportionately high $\serr$ expression size of 110 for clarity.  The average total size of the examples in the benchmarks was $15.91 \pm 17.18$\footnote{We use $a\pm s$ to denote an average $a$ with standard deviation $s$.}, with the maximum size of $116$. The average sizes for the individual expressions of the examples were:
i) $\scmd$: $2.38 \pm 1.01$ with maximum of $6$, ii) $\serr$: $10.12 \pm 16.85$ with a maximum of $110$, and iii) $\sfix$: $3.41 \pm 1.55$ with a maximum of $7$.

\paragraph{Distribution of rule matching}
For each set of example provided for an individual rule, we isolated one example to measure the accuracy of the tool.
All the test examples were correctly described by at least one of the synthesized rules.
For the majority of the test cases, there was exactly one rule which matched both the command and error message. The remaining $12$ test cases which matched against multiple rules came from collections of example sets which represented different fixes of the same command and error messages.\\

\begin{centering}
\begin{tabular} {l r}
  Total test cases & 81\\
  \hline
  Test cases matched by one rule & 69\\
  Test cases matched by multiple rules & 12\\
\end{tabular}\\\vspace{3mm}
\end{centering}

\begin{figure}
\centering
\includegraphics[scale=0.31]{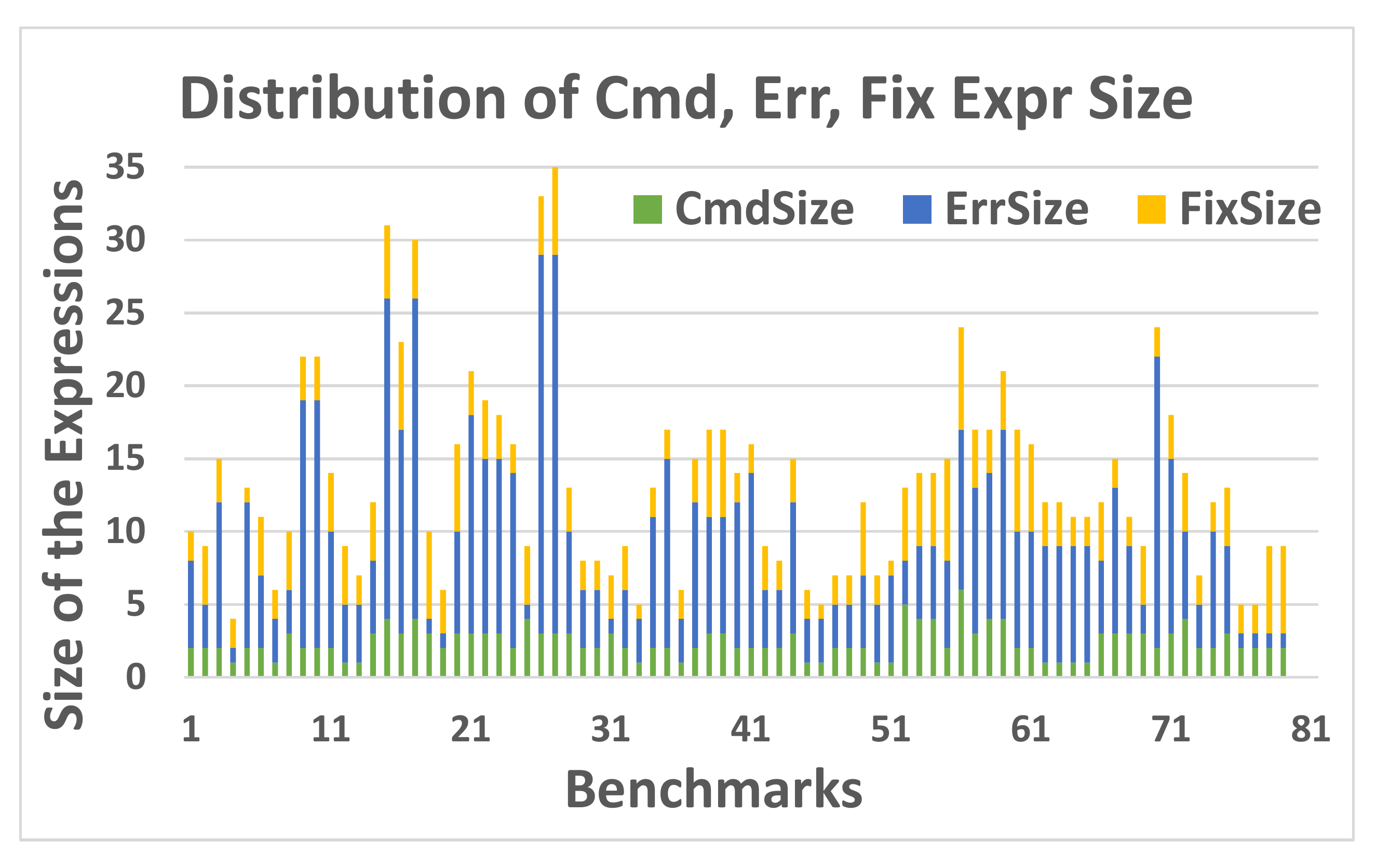}
\caption{The distribution of benchmarks in terms of individual sizes of $\scmd$, $\serr$, and $\sfix$ expressions in the examples.}
\figlabel{indsizebenchmarks}
\end{figure}

\paragraph{Ranking}
Consistent with our hypothesis in Section~\ref{sec:ranking1}, a diverse set of examples was sufficient for eliminating spurious restrictions and substring expressions. In every test case, the rule chosen by our ranking policy was capable of correcting all test cases presented. In practice, many rules still have several possible correct $\substrAppend$ expressions. However, this remaining ambiguity occurs because the same string can appear many times in the command and error message (e.g., the string Employee in the example in Section~\ref{sec:motiv:adding}).

\subsubsection{Quantitative evaluation}
We now report on the quantitative metrics of our synthesis algorithm.
In this section we only report data for the 81 benchmarks
for which \thefaq can successfully synthesize a \langc rule.

\begin{figure}
\centering
\includegraphics[scale=0.57]{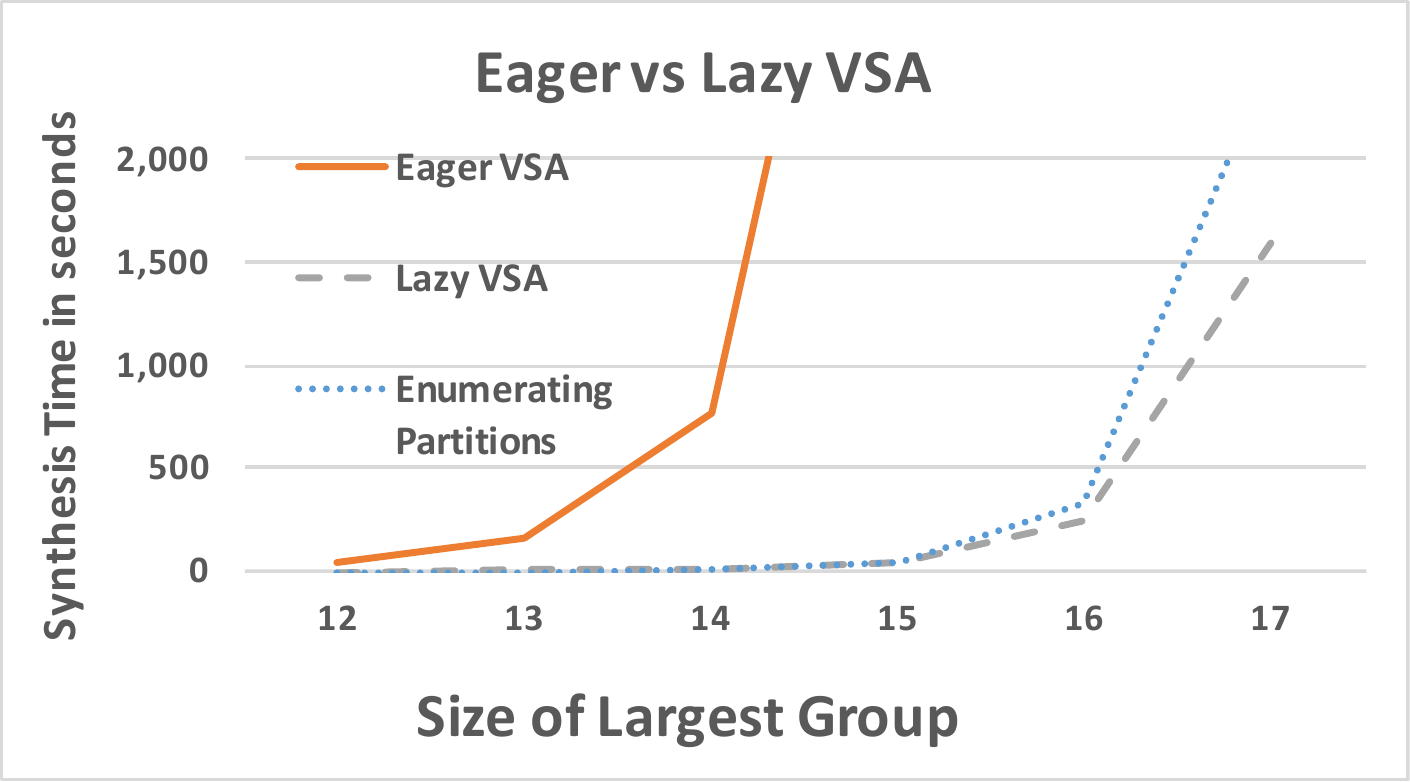}
\caption{Synthesis times for different benchmarks for the lazy and non-lazy rule representations.}
\figlabel{lazyvsnonlazy}
\end{figure}

\paragraph{Evaluation of lazy VSA synthesis time}
In \figref{lazyvsnonlazy}, we show the time taken to partition and synthesize \langc rules for the 81 benchmarks, using both the lazy and a non-lazy rule representation, as the number of examples per benchmark increases. The non-lazy representation always considers match and fix expressions as variables, rather than initially starting with constants.

To test the performance of the lazy and non-lazy representations as the size of the input set increases, we iteratively increase the size of the training set. For each test, we add a single example to one of the benchmarks and then attempt synthesis. We plot the synthesis time with respect to the largest set of examples for which we must enumerate possible partitions until we successfully synthesize rules. To understand the performance overhead induced by synthesis, we also evaluate a version of the algorithm which enumerates partitions without performing synthesis. For each algorithm, we iteratively increased the training set size until the algorithm reached a $2,000$ second timeout.

The non-lazy VSA incurs a significant overhead, and scales much worse than the lazy version, reaching the timeout when the largest set has 14 examples. The lazy VSA, in contrast, is much closer to the optimum; the synthesis time is negligible compared to the inherent cost of enumerating all partitions of a set. In fact, the lazy synthesis actually completes faster than exhaustive enumeration. This is reasonable, as the first partitioning which yields a successful \langc rule for all subsets tends to be somewhere near the middle of the enumeration, and thus does not incur the cost of enumerating the remainder of the search space.
In summary, the lazy VSA strictly outperforms non-lazy VSA and can handle much larger sets of examples.

\paragraph{\textup{$\substrAppend$} expression in synthesized rules}

The distribution of $\fconstStr$ and $\substrAppend$ expressions in the synthesized \langc rules is shown in \figref{learntdistribution}.
The output components of the synthesized rule contain on average $29.01\% \pm 24.4\%$ $\substrAppend$ expressions.
Concretely, a synthesized rule contains on average $0.91\pm 0.76$ $\substrAppend$ expressions.

\begin{figure}
\centering
\includegraphics[scale=0.31]{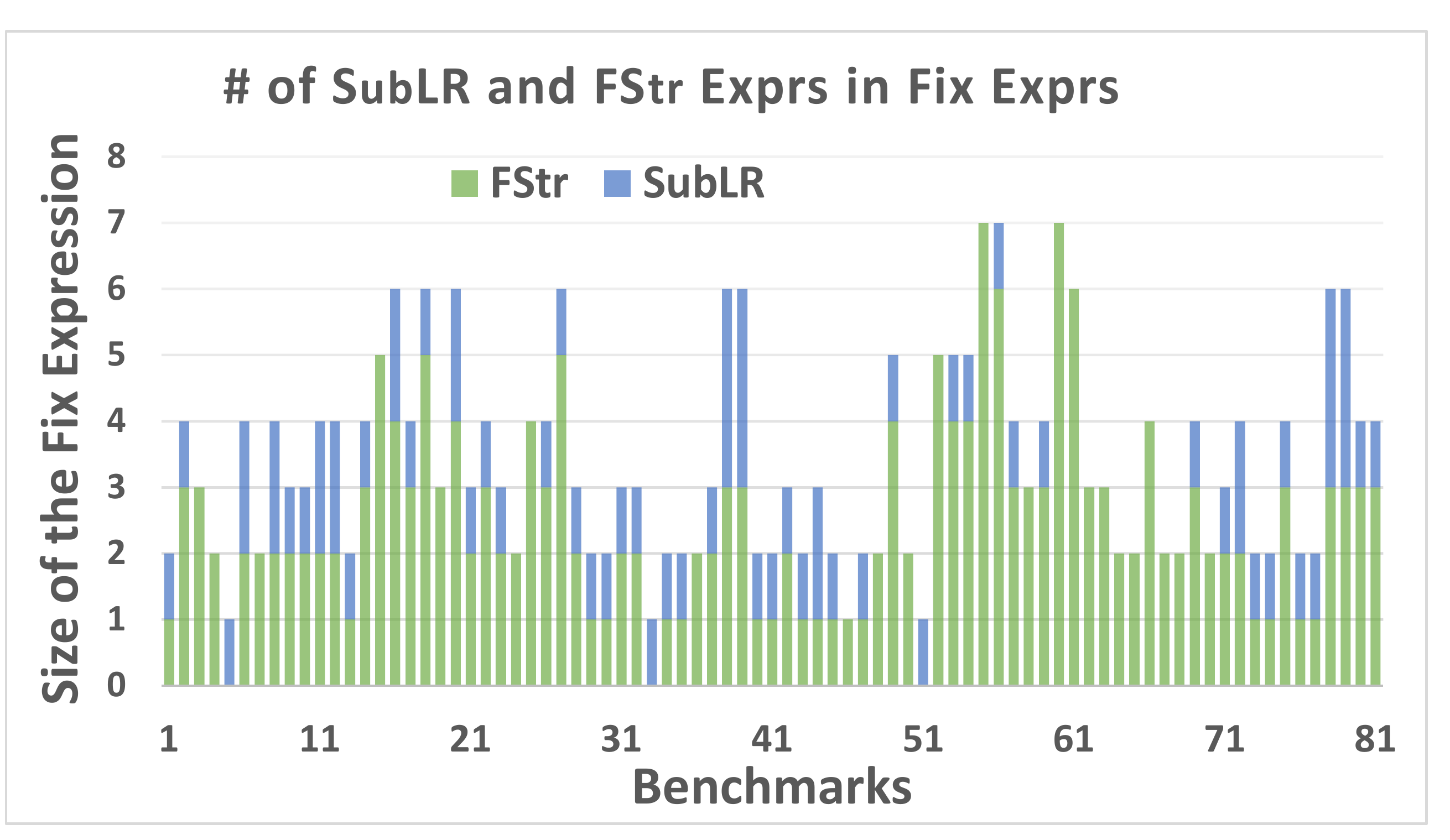}
\caption{The distribution of $\fconstStr$ and $\substrAppend$ expressions in the final synthesized repair expression.}
\figlabel{learntdistribution}
\end{figure}

\paragraph{Synthesis time vs. number of \textup{$\substrAppend$} expressions}
The synthesis time for different numbers of $\substrAppend$ expressions in the repair rule is shown in \figref{learningbyfixfunc}. As expected, the benchmarks that do not contain $\substrAppend$ expressions take negligible time. The benchmarks involving two $\substrAppend$ expressions on average require more time than the benchmarks with a single $\substrAppend$ expression. Interestingly, the benchmarks with 3 $\substrAppend$ expressions take lesser time than the benchmarks with 2 $\substrAppend$ expressions. A possible explanation for this behavior is that the complexity of substring extraction tasks for these benchmarks is relatively simpler (e.g. identity) than the benchmarks with 2 $\substrAppend$ expressions.

\begin{figure}
\centering
\includegraphics[scale=0.31]{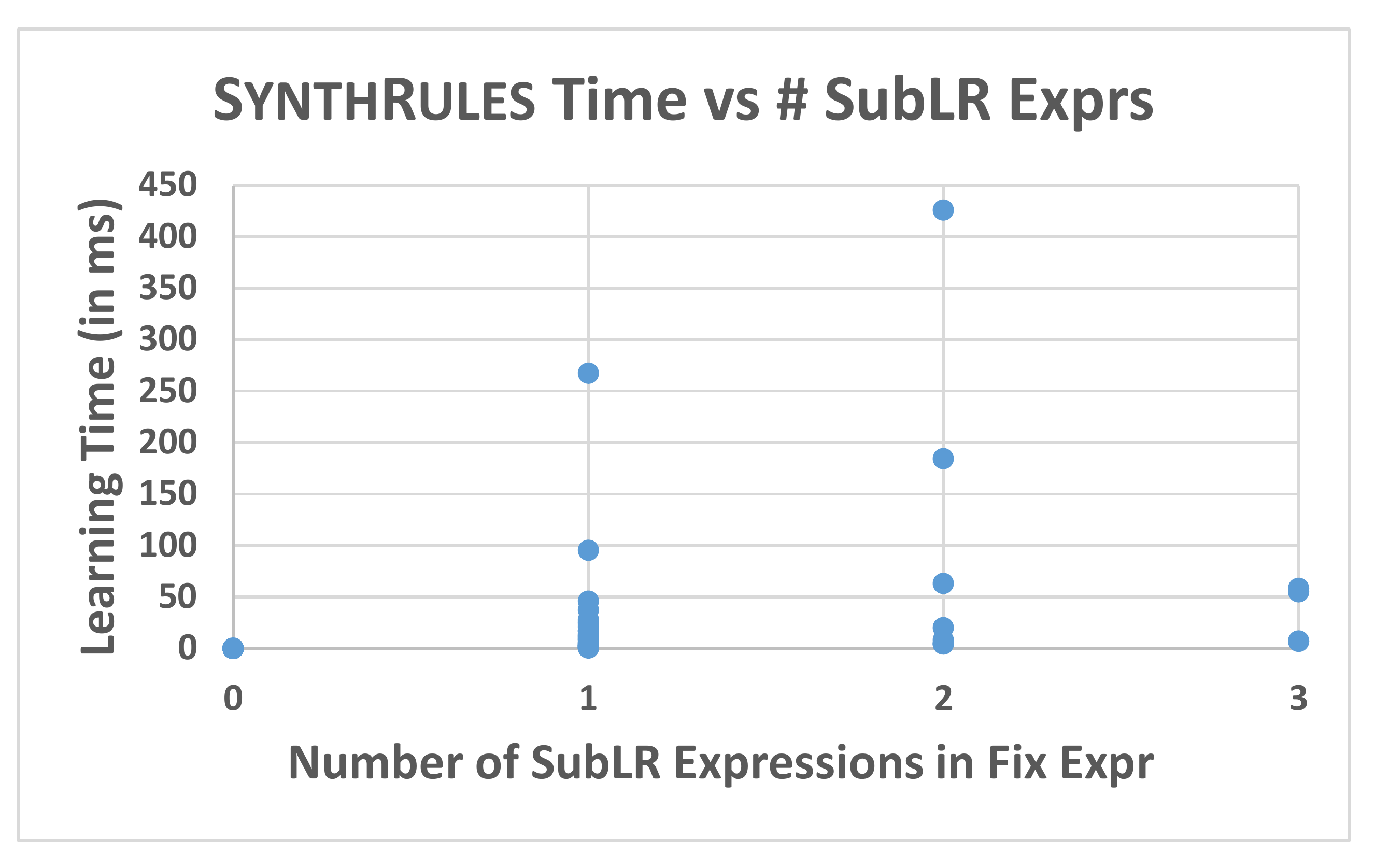}
\caption{Synthesis times for varying number of $\substrAppend$ expressions in the repair rule.}
\figlabel{learningbyfixfunc}
\end{figure}

\paragraph{Scalability of synthesis algorithm with example size} Since all real-world examples we collected are relatively of small size
(with maximum size of 116 space-separated strings), we evaluate the scalability of the \textsc{SynthRules} algorithm by creating artificial examples of increasing sizes. We create these artificial examples by repeating the $\scmd$, $\serr$, and $\sfix$ commands multiple times for the example shown in Section~\ref{sec:motiv:extractingsubs}. The synthesis times for increasing size of examples is shown in \figref{scalability}. From the graph, we observe that the synthesis times scale in a quadratic fashion with respect to the example size.

\begin{figure}
\centering
\includegraphics[scale=0.31]{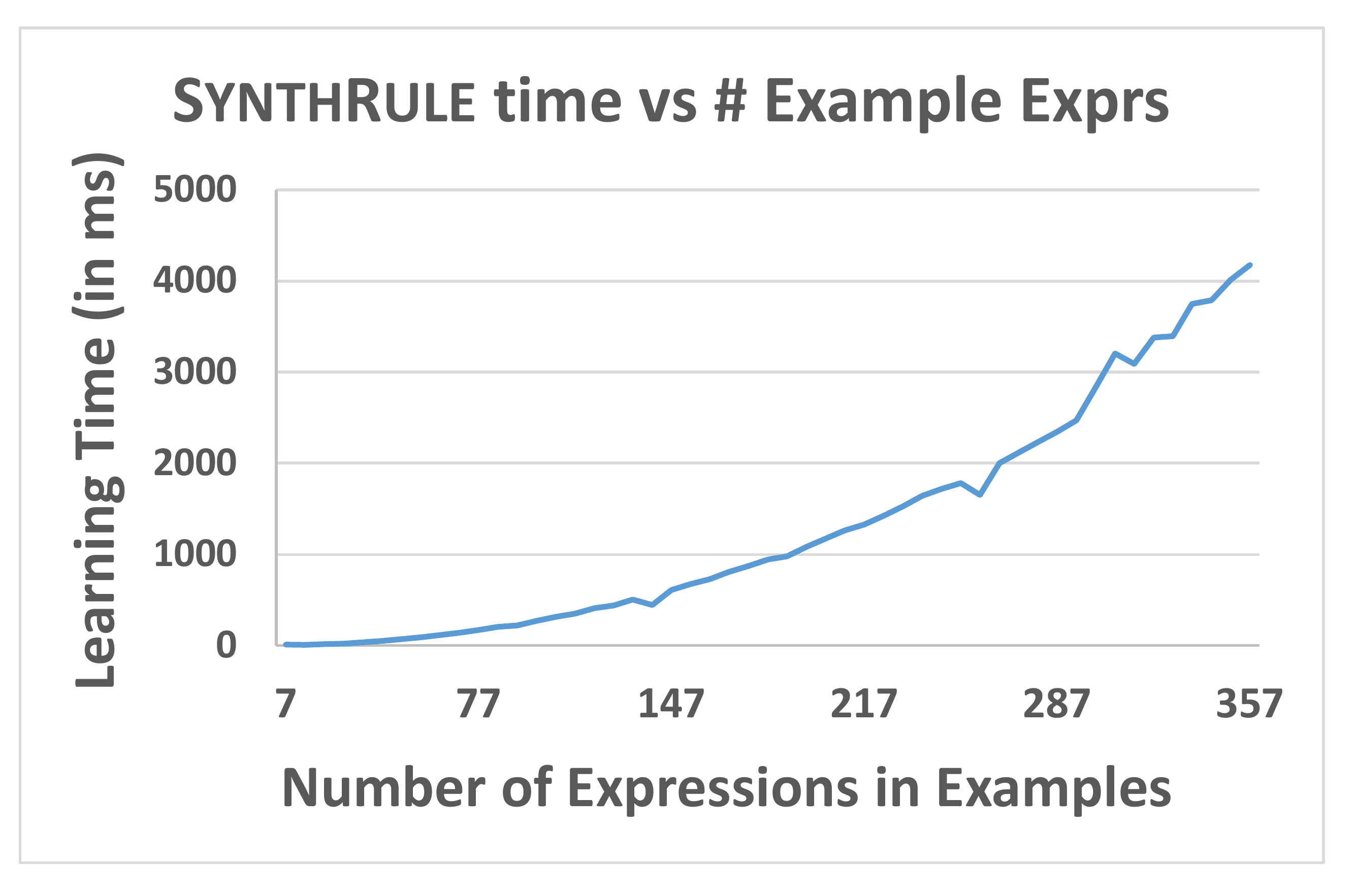}
\caption{Synthesis times with increasing size of examples.}
\figlabel{scalability}
\end{figure}

%% file: limitations.tex
\section{Limitations}
\label{sec:limitations}

We showed that the language $\langc$ is able to express many real-world command line repair rules
and that these rules can be synthesized using few examples.
We now present some limitations of our approach, in particular with respect to the
11 $\thefuck$ rules that $\langc$ could not describe.
\begin{description}		
\item[Complex patterns] Two rules were checking complex properties of the input
that \langc cannot capture. 
For example, \langc cannot check whether the error message contains some special character. 
\langc's conditional matching is limited to whole string or prefix/suffix matching, and thus cannot check if e.g. a file name contains a non-unicode whitespace character. All character relative logic occurs in the substring generation after input matching. \langc also cannot check whether some string in the input command is repeated more than once.

\item[Context-dependencies] Eight rules had hard-coded operations that were searching some context (the file system, a configuration file, etc.) for specific strings to complete the output.
\langc only receives as inputs the command and the error message, and the rules currently cannot use any context.


\end{description}

%% file: related.tex
\section{Related work}
\label{sec:related} 

\paragraph{Version-space algebra for synthesis} The concept of Version-space algebra (VSA) was first introduced by Mitchell~\cite{Mitchell82} in the context of machine learning and was later used by Lau et al. to learn programs from demonstrations in a Programming By Examples/Demonstrations system called SmartEdit~\cite{smartedit}. It has since been used for many PBE systems from various domains including syntactic string transformations in FlashFill~\cite{popl11}, table transformations~\cite{HarrisG11,vldb12}, number transformations~\cite{cav12}, text extraction from semi-structured text files in FlashExtract~\cite{flashextract}, and transformation of semi-structured spreadsheets to relational tables in FlashRelate~\cite{flashrelate}. Our synthesis algorithm also uses VSA to succinctly represent a large set of conforming expressions. However, in contrast to previous approaches that represent all conforming expressions concretely and then use intersection for refinement, our synthesis algorithm maintains a lazy representation of rules and concretizes the choices on demand in a lazy fashion only when it is needed. Moreover, our careful design of DSL operators and the corresponding VSA in $\thefaq$ lead to a polynomial time synthesis algorithm unlike most previous approaches that have exponential time synthesis algorithms.

In particular, it is illustrative to compare the FlashFill DSL with $\langc$. While, like $\langc$, FlashFill synthesizes string manipulations from input-output examples, specific performance properties make it less suitable for large scale learning from large sample sets. Prior to developing $\thefaq$, we evaluated the possibility of simply using the FlashFill algorithm as-is for our purpose of learning command repair rules. Early empirical results indicated that the off-the-shelf algorithm scaled poorly as the error messages increased in length, which was a common occurrence for our benchmarks. Moreover, other limitations of no offset operator in position expressions and support for finite hard-coded regular expression tokens made FlashFill unsuitable for learning $\substrAppend$ expressions.

We isolated several theoretical properties of FlashFill's key operators which yielded poor performance on large inputs. In particular, the binary concatenation operator over arbitrary substrings of the entire input string induces a DAG structure for the symbolic representation of programs. More explicitly, given an example output string $S$, there exists a node $n_p$ for each position $p$ in $S$. An edge from $n_p$ to $n_p'$, $p < p'$ represents the substring $S[p:p']$. Each edge is labelled with the set $F$ of functions over the example inputs which yield the substring. Thus, a path from $n_0$ to $n_{|S|}$ represents some concatenation of the output of several string operations which yields the desired output. Given a DAG $D$ consistent with a set of examples $E$, FlashFill incorporates a new example $e$ represented by DAG $D'$ by taking the cartesian product of the vertices of $D$ and $D'$, to construct a new DAG $D''$. An edge with label set $F''$ in the new DAG represents a set of functions which were part of a correct program for the examples $E$, and also map from the inputs of $e$ to a substring of $e$'s output. The iterated cartesian products yield time complexity exponential in the number of examples.

$\langc$, in contrast, posesses unary string operations constrained to specific variable terms. $\langc$'s unary $\substrAppend$ operator yields a language that is disjoint from FlashFill with concatenation removed; we obtain a language expressive enough for a large set of practical command repair transformations isolated from real use cases, while dramatically improving worst-case performance. The constrained nature of the program representation lets the synthesis algorithm eliminate programs inconsistent with a new example without directly computing the intersection of two sets of candidate programs, ensuring polynomial-time performance even in the worst case. Moreover, $\langc$ also allows for repair transformations that require arbitrary offsets from a regex match, which are not expressible in the FlashFill DSL.

\paragraph{Programming by Examples (PBE)} PBE has been an active research area in the AI and HCI communities from a long time~\cite{lieberman2001your}.  In addition to VSA-based data wrangling~\cite{cacm12}, PBE techniques have recently been developed for various domains including interactive synthesis of parsers~\cite{LeungSL15}, synthesis of recursive functional programs over algebraic data types~\cite{FeserCD15,OseraZ15}, synthesizing sequence of program refactorings~\cite{RaychevSSV13}, imperative data structure manipulations~\cite{storyboardfse}, and network policies~\cite{YuanAL14}. Our technique also learns repair rules from few input-output examples of buggy and fixed commands, but both our problem domain of learning command repairs and the learning techniques of using lazy VSA are quite different from these PBE systems.

\paragraph{Program repair}

Research in automated program repair focuses on automatically changing
incorrect programs to make them meet a desired specification~\cite{GouForWei13}.
The main challenge is to efficiently search the space of all programs
to find one that behaves correctly.
The most prominent search techniques are enumerative or data-driven.
GenProg uses genetic programming
to repeatedly alter the incorrect program in the hope to make it correct~\cite{LeGDewForWei12}. 
Data-driven approaches use the large amount of code that is publicly available online
to synthesize likely changes to the input program~\cite{RayVecYah14,YakElmNevOuzIly11}. Prophet~\cite{prophet} is a patch generation system that learns a probabilistic application-independent model of correct code from a set of successful human patches. Unlike these techniques that learn a global model of code repair across different applications, our technique learns command-specific repairs by observing how expert
users fix their buggy commands --- i.e., from both the incorrect command the user started with (together with the error message)
and the correct command she wrote as a fix.

\paragraph{Crowdsourced Repairs}
\emph{HelpMeOut} is a social recommender system that 
helps novice users facing programming errors by showing them examples 
of how other programmers
have corrected similar errors~\cite{HarMacBraKle10}.
While \emph{HelpMeOut} can show examples of similar fixes it does not
concretely show
the user how the code should be corrected.
This aspect is the major difference between \emph{HelpMeOut} and \thefaq.

\thefuck provides a Python interface for command substitution and repair rules, and it requires a degree of language and tool-specific knowledge that may not be accessible to command line novices, particularly if non-trivial substring operations are required to derive the desired command. Much like FlashFill, we aim to emulate the workflow of non-technical users communicating with experts on web forums. For a beginner learning the command line, Python string manipulations are likely a fairly challenging task, and the cost of an incorrectly transformed shell command is potentially catastrophic. In such a situation where a non-expert desires a new \thefuck rule, such a user may provide an example of several command/error pairs, and the desired fix for each, from which an expert would write the desired Python code. NoFAQ shortens this loop by moving the fix synthesis into a polynomial time algorithm on the user's machine.

\paragraph{Rule learning}
Rules provide a simple way to represent programmers actions and in general any type of data transformation.
Rule learning has been extensively investigated in classical machine learning and data mining~\cite{FurGamLav12}. 
The goal of rule learning is to discover and mine rules describing interesting relations 
appearing in data.
Common concept classes for describing rules are Horn clauses or association rules~\cite{PiaSha91}.
The approach presented in this paper differs from rule learning in two aspects:
1) the rules are expressed in a complex concept class and are hard to learn --- i.e., \langc programs;
2) the examples are given by a teacher that has in mind a target rule.
In the future we plan to build a system that uses rule learning techniques to mine \langc rules
from unsupervised data.

\paragraph{Program synthesis}
There has been a resurgence in Program Synthesis research in recent years~\cite{sygus}. In addition to examples (as described above), there have been several techniques developed for handling other forms of specifications such as partial programs~\cite{sketchthesis,solar2005bitstreaming}, reference implementations~\cite{Schkufza0A13}, and concrete traces~\cite{UdupaRDMMA13}. While these specification mechanisms have been found to be useful in several domains, we believe examples are the most natural mechanism for specifying command line repairs especially for beginners. There is also a recent movement towards using data-driven techniques for synthesis~\cite{RaychevVK15}, e.g. the PLINY  project (\url{http://pliny.rice.edu/index.html}). In future, we envision our system to also make use of large number of examples of buggy commands and their corresponding repairs to learn a big database of \langc  rules.

%% file: conclusion.tex
\section{Conclusion and Future Directions}
We presented a tool \thefaq that suggests possible fixes to common buggy commands by learning from examples of how experts fix such issues. Our language design walks a fine line between expressivity and performance:
 by careful choice of unary operators over pre-defined variables, and exclusion of arbitrary substring operations, we avoid exponential-time worst case performance, while still maintaining a useful degree of functionality.
 \thefaq was able to instantly synthesize 85\% of the 
rules appearing in the popular repair tool \thefuck and 16
other rules from online help forums.
Although $\thefaq$ tool is aimed towards repairing commands, we believe our novel combination of synthesis and rule-based program repair is 
quite general and is applicable in many other domains as well. We plan to
to apply this methodology to more complex tasks, such
as correcting syntax errors in source code, applying code optimization, and editing configuration files. In the
 future, we hope to create a tool which can take large command histories from expert users and quickly derive rules, as well as synthesize new rules online as experts use the shell.

%% file: appendix.tex
\section{Proofs of Theorems 1, 2, and 3}

We first define a notion of completeness for a symbolic rule. Intuitively a symbolic rule 
has to summarize all possible correct rules to be complete.
\begin{definition}[Command-string completeness]
\label{def:commandcomp}
Let 
$$r=\sfixRule{cmd}{err}{fixes}$$ 
be a symbolic rule such that
$cmd=[c_1,\ldots, c_a]$,
\linebreak $err=[m_1,\ldots, m_b]$,
$fixes=[f_1,\ldots, f_c]$,
 and
$E=[e_1,\ldots, e_n]$ a sequence of examples.
We say that $cmd$ is complete for $E$ and produces variables $V_1$, 
$CompC(cmd,E,V_1)$ 
iff for every $1\leq i\leq a$:
\begin{itemize}
\item $c_i=\mbox{Str}(s)$ (for some $s$) iff every example $(\scmd,\serr,\sfix)$ in $E$ is
such that $\scmd=[s_1,\ldots,s_k]$, for some $k$, and $s_j=s$.
\item $c_i=\mbox{Var-Match}(j,l,r)$ iff $i\in V$, $i=j$, and there exists two examples 
$(\scmd^1,\serr^1,\sfix^1)$ and $(\scmd^2,\serr^2,\sfix^2)$ in $E$ such that $\scmd^1=[s_1,\ldots,s_k]$,  
$\scmd^2=[s_1',\ldots,s_k']$, for some $k$,
  and $s_j\neq s_j'$.
\item If $c_i=\mbox{Var-Match}(j,l,r)$, then for every example \linebreak $(\scmd,\serr,\sfix)$ in $E$ where $\scmd = [s_1, \ldots, s_k]$, $l$ is a prefix of $s_j$. Moreover, $l$ is the longest such prefix.
\item If $c_i=\mbox{Var-Match}(j,l,r)$, then for every example \linebreak $(\scmd,\serr,\sfix)$ in $E$ where $\scmd = [s_1, \ldots, s_k]$, $r$ is a suffix of $s_j$. Moreover, $r$ is the longest such suffix.
\end{itemize}
\end{definition}
\begin{definition}[Error-string completeness]
\label{def:errcomp}
Analogously, we say that $err$ is complete for $E$ and produces variables $V_2$, 
\linebreak $CompE(err,E,V_2)$, iff
for every $1\leq i\leq b$:
\begin{itemize}
\item $m_i=Str(s)$ (for some $s$) iff every example $(\scmd,\serr,\sfix)$ in $E$ is
such that $\scmd=[s_1,\ldots,s_k]$, for some $k$, and $s_j=s$.
\item $m_i=\mbox{Var-Match}(j,l,r)$ iff $i\in V$, $i+a=j$, and there exists two examples 
$(\scmd^1,\serr^1,\sfix^1)$ and $(\scmd^2,\serr^2,\sfix^2)$ in $E$ such that $\serr^1=[s_1,\ldots,s_k]$,  
$\serr^2=[s_1',\ldots,s_k']$, for some $k$,
  and $s_i\neq s_i'$.
\item If $m_i=\mbox{Var-Match}(j,l,r)$, and $i+a=j$, then for every example $(\scmd,\serr,\sfix)$ in $E$ where $\serr = [s_1, \ldots, s_k]$, $l$ is a prefix of $s_i$. Moreover, $l$ is the longest such prefix.
\item If $m_i=\mbox{Var-Match}(j,l,r)$, and $i+a=j$ then for every example $(\scmd,\serr,\sfix)$ in $E$ where $\serr = [s_1, \ldots, s_k]$, $r$ is a suffix of $s_i$. Moreover, $r$ is the longest such suffix.
\end{itemize}
\end{definition}

\begin{definition}[Input completeness]
  \label{def:inputcomp}
If both $CompC(cmd,E,V_1)$ and
$CompE(err,E,V_2)$ hold we say
that $cmd$ and $err$ are complete for $E$ and produce variables $V_1\cup V_2$,\linebreak
$Comp(cmd, err,E,V_1\cup V_2)$.
\end{definition}

\begin{definition}[Partial $f_i$-completeness]
  \label{def:partialficomp}
We say that $f_i$ is 
partially complete, $PCompFi(f_i,E,i)$,
with respect to $E$ if the following condition holds:
$f_i=[\fconstStr(s)]$ (for some $s$) iff every example $(\scmd,\serr,\sfix)$ in $E$ is
such that $\sfix=[s_1,\ldots,s_k]$, for some $k$, and $s_i=s$.
If for every $1\leq i\leq c$, $PCompFi(f_i,E,i)$ holds, then we say that 
that $fixes$ is partially complete with respect to $V$, $PCompF(fixes,E)$.
\end{definition}

\begin{definition}[$f_i$-completeness]
  \label{def:ficomp}
If there exists a $V=\{i_1,\ldots,i_j\}$ such that 
$Comp(cmd,err,E,V)$ we say that $f_i$ is complete 
with respect to $V$ at position $i$, $CompFi(f_i,E,V,i)$, iff:
\begin{itemize}
\item $PCompFi(f_i,E,i)$.

\item $f_i=[t_1,\ldots,t_m]$ such that for all $ind$, $t_{ind}\neq \fconstStr(s)$ (for any $s$) iff
the following properties hold.
\begin{itemize}
\item For every $ind\leq m$ and $l\leq k$,
$\semanticsfun{t_{ind}}_{\state_l}=\sfix^l[i]$ (where for all $y\leq j$, $\state_l(i_y)=(\scmd^l @ \serr^l)[i_y]$).

\item 
If there exists a $\substrAppend$ function $t$ such that 
for every $ind\leq n$ and $l\leq k$, ,
$\semanticsfun{t}_{\state_l}=\sfix^l[i]$ (where for all $y\leq j$, $\state_l(i_y)=(\scmd^l @ \serr^l)[i_y]$), then 
there exists $x\leq m$ such that $t_x=t$.
\end{itemize}
\end{itemize}
\end{definition}
\begin{definition}[Fix completeness]
If for every $1\leq i\leq c$, $CompFi(f_i,E,V,i)$ holds, then we say that 
that $fixes$ is complete with respect to $V$, $CompF(fixes,E,V)$.
\end{definition}
\begin{definition}[Rule completeness]
 \label{def:completerule}
If there is a $V$ such that \linebreak 
$Comp(cmd,err,E,V)$ and $CompF(fixes,E,V)$, we say that $r$ is complete for $E$, $CompR(r,E)$.
\end{definition}
Notice that for any permutation $E'$ of $E$ $CompR(r,E)$ iff $CompR(r,E')$.

\begin{proposition}[$CompR$ and $\conc$]
\label{prop:deftocon}
Let
$E$ a sequence of examples and 
$r=\sfixRule{cmd}{err}{fixes}$ be a symbolic rule.
If $CompR(r,E)$ then every concrete rule $r'\in \conc(r)$ is consistent with any example in $E$.
Moreover, for every (non-symbolic) rule 
$$r_1=\fixRule{cmd}{err}{[t_1,\ldots,t_n]}$$ 
consistent with 
$E=[(\scmd^1,\serr^1,\sfix^1),\ldots,(\scmd^n,\serr^n,\sfix^n)]$ the following is true:
if for every $i$, $t_i\neq \fconstStr(s)$ iff for some $j_1,j_2\leq n$
$\sfix^{j_1}[i]\neq \sfix^{j_2}[i]$, then
$r_1\in \conc(r)$.
\end{proposition}
\begin{proof}
Immediate from Definition~\ref{def:completerule}
and the definition of $\conc$.

  More specifically, assume there is some $r' \in \conc(r)$ not consistent with an example $e \in E$. The only way it can be inconsistent with $\scmd$ or $\serr$ is if some match expression in $cmd$ or $err$ is a fixed string not equal to some string in $\scmd$ or $\serr$. Definition \ref{def:inputcomp} precludes this possibility.

  Similarly, it follows from Definition \ref{def:partialficomp} that if one of the fix expressions is a constant string, then it is consistent with every example fix. By Definition \ref{def:ficomp}, if it is a set of $\substrAppend$ expressions, then each one is consistent with every example.

\end{proof}

We now show that \textsc{SynthSubstring} and \textsc{SynthFix} have the intended behaviour
with respect to Definition~\ref{def:completerule}.
\begin{lemma}[Correctness of \textsc{SynthSubstring}]
\label{lem:syntsub}
Let \linebreak
$E=[(\scmd^1,\serr^1,\sfix^1),\ldots,(\scmd^n,\serr^n,\sfix^n)]$
be a sequence of examples where every $\sfix^v$ has length $m$, 
$V=\{i_1,\ldots,i_j\}$ be a set of variables, and
$i\leq m$ be an index in the output fix.
If there exist $a$ and $b$ such that $\sfix^a[i]\neq \sfix^b[i]$, then
$\textsc{SynthSubstring}(E,V,i)=S$ iff
$Comp(S,E,V,i)$.
\end{lemma}
\begin{proof}
This lemma states that  \textsc{SynthSubstring} returns all the $\substrAppend$ functions  consistent
with  the given examples.
Again the proof is by induction on the length of $E$. The base case $|E|=1$
follows from the definition of \textsc{AllSubstrings}.
The inductive step is also trivial: \textsc{SynthSubstring} simply runs all the functions computed so far
on the added example and filters out those that are not consistent with it.
Since, by IH, the varaible $F$ was correct at the beginning of the loop, it remains correct.
Notice that the order of the examples does not matter.
\end{proof}

\begin{lemma}[Correctness of \textsc{SynthFix}]
\label{lem:corsynthfix}
Let
$S=[s_1,\ldots, s_n]$ be a list of strings,
$T=[t_1,\ldots, t_n]$ be a sequence of symbolic fix expressions,
$$E=[(\scmd^1,\serr^1,\sfix^1),\ldots,(\scmd^n,\serr^n,\sfix^n)]$$
be a sequence of examples where every $\sfix^v$ has length $m$, 
$e=(\scmd,\serr,\sfix)$ be an example such that for each $1\leq i \leq n$, $\sfix[i]=s_i$, and
$V$ be a set of variables.
If $PCompF(T, E)$, then 
$\textsc{SynthFix}(S,T,e::E,V)=fixes$ iff
$CompF(fixes, e::E, V)$.
\end{lemma}
\begin{proof}
Immediate by induction on $m$ and by using 
Lemma~\ref{lem:syntsub} at each step.\\
  As $\textsc{SynthFix}$ iterates over each string in the example fix, and compares it to the corresponding fix expression, we can see that the only interesting case of the proof is when $t_i$ is either a fixed string not equal to the correpsonding $\sfix^{n+1}[i]$ in $e$, or $t_i$ is a set of $\substrAppend$ expressions. In this case, correctness follows directly from Lemma~\ref{lem:syntsub}, as $\textsc{SynthFix}$ calls $\textsc{SynthSubstring}$ to refine the set of substring operation to those generating $\sfix^{n+1}[i]$ as well as all other $\sfix[i]$ values.

\end{proof}

Next we show the correctness of \textsc{RefineRule} and \textsc{FindVariables}.
\begin{lemma}[Correctness of \textsc{FindVariables}]
\label{lem:corrfindvars}
Let \linebreak
$E=[(\scmd^1,\serr^1,\sfix^1),\ldots,(\scmd^n,\serr^n,\sfix^n)]$ be a sequence of examples, and
$cmd$ and $err$ be two sequences of match expressions of lengths
$|\scmd^i|$ and $|\serr^i|$ respectively.
If there exists two sets of variables $V_1$ and $V_2$ such that 
$CompC(cmd,E[1::(n-1)],V_1)$ and $CompE(err,E[1::(n-1)],V_2)$,
then
\begin{itemize}
\item if $(cmd',V_1')=\textsc{FindVariables}(\scmd^n,cmd,0)$, then\linebreak
$CompC(cmd',E,V_1')$.
\item if $(err',V_2')=\textsc{FindVariables}(\serr^n,err,|\scmd^i|)$, then\linebreak
$CompE(err',E,V_2')$.
\end{itemize}
\end{lemma}
\begin{proof}
The two statements can be proved separately but the proof is identical. 
The proofs are both by induction on the length of the first argument of
\textsc{FindVariables} and are very simple case analysis following
Definitions~\ref{def:commandcomp} and~\ref{def:errcomp}.

W.L.O.G, consider the case for $cmd$. \linebreak 
Over a run of $\textsc{FindVariables}$, a $\constStr(s)$ in $cmd$ is only changed if, at the $n$th iteration of the loop, $cmd[i]$  is $\constStr(S)$ where $S$ is not equal to $\scmd^n[i]$. In that case \textsc{FindVariables} introduces a new variable match expression $\varmatch(i,l,r)$ where $l$ and $r$ are (respectively) the longest shared prefix and suffix of $S$ and $\scmd^n[i]$. By the I.H, $cmd$ was command complete. It follows directly from Definition~\ref{def:commandcomp} that command completeness continues to hold.

A $\varmatch(i,l,r)$ in $cmd$ is only changed if, at iteration $n$, the longest common prefix (suffix), $x$, shared by $\scmd^n[i]$ and $l$ ($r$) is not equal to $l$ ($r$). In this case $\varmatch(i,l,r)$ is replaced by $\varmatch(i,l',r)$ ($\varmatch(i,l,r')$), where $r'$ ($l'$) is equal to the common prefix (suffix) $x$. By the I.H, $l$ was the longest prefix (suffix) shared by the first $n-1$ values of $\scmd^k[i]$. Clearly, $x = r'$ ($l'$) is the longest prefix (suffix) shared by the first $n$ values of $\scmd^k[i]$. Moreover, by the I.H. $cmd$ satisfied all other criteria for command completeness. Thus, it follows that command completeness continues to hold. It is clear that determining the shorteset prefix and suffix shared by every $\scmd^k[i]$ does not depend on the order in which the examples are presented.

\end{proof}

\begin{lemma}[Correctness of \textsc{RefineRule}]
\label{lem:refinerule}
Let $r$ be a symbolic rule,
$E$ be a sequence of examples, 
 $e$ be an example, and \linebreak
 $r'=\textsc{RefineRule}(r,E,e)$.
 If $CompR(r,E)$ then $CompR(r',e::E)$.
\end{lemma}
\begin{proof}
Immediate from Definition~\ref{def:completerule} and Lemmas~\ref{lem:corsynthfix} and~\ref{lem:corrfindvars}.
\end{proof}

\begin{lemma}[Correctness of \textsc{SynthRules}]
\label{lem:synthrules}
Let
$E$ be a sequence of examples.
If $r'=\textsc{SynthRules}(E)$ then $CompR(r,E)$.
\end{lemma}
\begin{proof}
By induction on the length of $E$.
Case $|E|=1$ and $E=[e]$: the result of $\textsc{ConstRule}(e)$ clearly satisfies $CompR(r,E)$.
The inductive step follows from Lemma~\ref{lem:refinerule}.
\end{proof}

We can now conclude the proofs of Theorems 1, 2, and 3.
Theorem 1 and 2 follow from the order invariance of $CompR(r,E)$ and from Proposition~\ref{prop:deftocon}
and Lemmas~\ref{lem:synthrules}.

\paragraph{Proof of Theorem 3}
Let 
$$E=[(\scmd^1,\serr^1,\sfix^1),\ldots,(\scmd^n,\serr^n,\sfix^n)]$$
be a sequence of examples,
$$\textsc{SynthRules}(E)=r=\sfixRule{\origCmd}{\errMsg}{fixes}$$ 
be the symbolic rule for $E$ and
$$r'=\fixRule{\origCmd'}{\errMsg'}{\fixCmd'}$$ 
be a concrete rule that is consistent with all the examples in $E$.
If $r'$ belongs to $\conc(r)$ then we are done.
Assume it doesn't. 

\textbf{Case 1.} $\origCmd=cmd'$ and $err=err'$. Then by Proposition~\ref{prop:deftocon} there exists some $t_{i_1},\ldots,t_{i_k}$ such that 
each $t_{i_v}$ is of the form $\substrAppend(\pl^v,\pr^v,\pref^v,\suf^v,\varexpr(j_v))$,
and for all $j\leq n$
$\sfix^j[i_v]=s$ for some $s$ (i.e., the output is a function of the input, but 
all examples can be captured using a constant output).
In this case it is enough to create a new example $e'$ starting from any example $(\scmd,\serr,\sfix)\in E$ 
where for each $i_v$ we
modify $s_{j_v}=(\scmd @ \serr)[j_v]$ so that 
$\substrAppend(\pl^v,\pr^v,\pref^v,\suf^v,\varexpr(j_v))$ now returns a value different from the previous one (this can be done
by simply adding a new character between $\pl^v$ and $\pr^v$).
$\sfix$ is replaced by the result of applying $r'$ to the modified input.

\textbf{Case 2.}
$\origCmd @ err\neq cmd' @ err'$.
This means that $r'$ uses more variables then $r$
(notice that from Definition~\ref{def:completerule} the set of variables used
by $r'$ is necessary).
This can be fixed by changing the input of any example $(\scmd,\serr,\sfix)\in E$.
For each variable $\varexpr(i)$ that is in $cmd' @ err'$ but not in $\origCmd @ err$
replace the string $(\scmd @ \serr)[i]=a_1\ldots a_n$ with $ba_1 \ldots a_n$
where $b$ is a symbol not appearing in $a_1\ldots a_n$.
$\sfix$ is replaced by the result of applying $r'$ to the modified input.
If the rule is still not in $\conc(r')$ we can then modify the example using the techniques from Case 1 since
now $\origCmd=cmd'$ and $err=err'$.

%% file: paper.bbl
\begin{thebibliography}{31}
\providecommand{\natexlab}[1]{#1}
\providecommand{\url}[1]{\texttt{#1}}
\expandafter\ifx\csname urlstyle\endcsname\relax
  \providecommand{\doi}[1]{doi: #1}\else
  \providecommand{\doi}{doi: \begingroup \urlstyle{rm}\Url}\fi

\bibitem[Alur et~al.(2015)Alur, Bod{\'{\i}}k, Dallal, Fisman, Garg, Juniwal,
  Kress{-}Gazit, Madhusudan, Martin, Raghothaman, Saha, Seshia, Singh,
  Solar{-}Lezama, Torlak, and Udupa]{sygus}
R.~Alur, R.~Bod{\'{\i}}k, E.~Dallal, D.~Fisman, P.~Garg, G.~Juniwal,
  H.~Kress{-}Gazit, P.~Madhusudan, M.~M.~K. Martin, M.~Raghothaman, S.~Saha,
  S.~A. Seshia, R.~Singh, A.~Solar{-}Lezama, E.~Torlak, and A.~Udupa.
\newblock Syntax-guided synthesis.
\newblock In \emph{Dependable Software Systems Engineering}, pages 1--25. 2015.

\bibitem[Barowy et~al.(2015)Barowy, Gulwani, Hart, and Zorn]{flashrelate}
D.~W. Barowy, S.~Gulwani, T.~Hart, and B.~Zorn.
\newblock Flashrelate: Extracting relational data from semi-structured
  spreadsheets using examples.
\newblock In \emph{Proceedings of the 36th ACM SIGPLAN Conference on
  Programming Language Design and Implementation}, PLDI 2015, pages 218--228,
  New York, NY, USA, 2015. ACM.
\newblock ISBN 978-1-4503-3468-6.
\newblock \doi{10.1145/2737924.2737952}.
\newblock URL \url{http://doi.acm.org/10.1145/2737924.2737952}.

\bibitem[Doane et~al.(1992)Doane, McNamara, Kintsch, Polson, and
  Clawson]{promptcomprehension}
S.~M. Doane, D.~S. McNamara, W.~Kintsch, P.~G. Polson, and D.~M. Clawson.
\newblock Prompt comprehension in unix command production.
\newblock \emph{Memory \& cognition}, 20\penalty0 (4):\penalty0 327--343, 1992.

\bibitem[Feser et~al.(2015)Feser, Chaudhuri, and Dillig]{FeserCD15}
J.~K. Feser, S.~Chaudhuri, and I.~Dillig.
\newblock Synthesizing data structure transformations from input-output
  examples.
\newblock In \emph{Proceedings of the 36th {ACM} {SIGPLAN} Conference on
  Programming Language Design and Implementation, Portland, OR, USA, June
  15-17, 2015}, pages 229--239, 2015.
\newblock \doi{10.1145/2737924.2737977}.
\newblock URL \url{http://doi.acm.org/10.1145/2737924.2737977}.

\bibitem[Furnkranz et~al.(2012)Furnkranz, Gamberger, and Lavrac]{FurGamLav12}
J.~Furnkranz, D.~Gamberger, and N.~Lavrac.
\newblock Rule learning in a nutshell.
\newblock In \emph{Foundations of Rule Learning}, Cognitive Technologies, pages
  19--55. Springer Berlin Heidelberg, 2012.
\newblock ISBN 978-3-540-75196-0.
\newblock \doi{10.1007/978-3-540-75197-7_2}.
\newblock URL \url{http://dx.doi.org/10.1007/978-3-540-75197-7_2}.

\bibitem[Goues et~al.(2013)Goues, Forrest, and Weimer]{GouForWei13}
C.~Goues, S.~Forrest, and W.~Weimer.
\newblock Current challenges in automatic software repair.
\newblock \emph{Software Quality Journal}, 21\penalty0 (3):\penalty0 421--443,
  Sept. 2013.
\newblock ISSN 0963-9314.
\newblock \doi{10.1007/s11219-013-9208-0}.
\newblock URL \url{http://dx.doi.org/10.1007/s11219-013-9208-0}.

\bibitem[Gulwani(2011)]{popl11}
S.~Gulwani.
\newblock Automating string processing in spreadsheets using input-output
  examples.
\newblock In \emph{Proceedings of the 38th Annual ACM SIGPLAN-SIGACT Symposium
  on Principles of Programming Languages}, POPL '11, pages 317--330, New York,
  NY, USA, 2011. ACM.
\newblock ISBN 978-1-4503-0490-0.
\newblock \doi{10.1145/1926385.1926423}.
\newblock URL \url{http://doi.acm.org/10.1145/1926385.1926423}.

\bibitem[Gulwani et~al.(2012)Gulwani, Harris, and Singh]{cacm12}
S.~Gulwani, W.~R. Harris, and R.~Singh.
\newblock Spreadsheet data manipulation using examples.
\newblock \emph{Commun. ACM}, 55\penalty0 (8):\penalty0 97--105, Aug. 2012.
\newblock ISSN 0001-0782.
\newblock \doi{10.1145/2240236.2240260}.
\newblock URL \url{http://doi.acm.org/10.1145/2240236.2240260}.

\bibitem[Harris and Gulwani(2011)]{HarrisG11}
W.~R. Harris and S.~Gulwani.
\newblock Spreadsheet table transformations from examples.
\newblock In \emph{Proceedings of the 32nd {ACM} {SIGPLAN} Conference on
  Programming Language Design and Implementation, {PLDI} 2011, San Jose, CA,
  USA, June 4-8, 2011}, pages 317--328, 2011.
\newblock \doi{10.1145/1993498.1993536}.
\newblock URL \url{http://doi.acm.org/10.1145/1993498.1993536}.

\bibitem[Hartmann et~al.(2010)Hartmann, MacDougall, Brandt, and
  Klemmer]{HarMacBraKle10}
B.~Hartmann, D.~MacDougall, J.~Brandt, and S.~R. Klemmer.
\newblock What would other programmers do: Suggesting solutions to error
  messages.
\newblock In \emph{Proceedings of the SIGCHI Conference on Human Factors in
  Computing Systems}, CHI '10, pages 1019--1028, New York, NY, USA, 2010. ACM.
\newblock ISBN 978-1-60558-929-9.
\newblock \doi{10.1145/1753326.1753478}.
\newblock URL \url{http://doi.acm.org/10.1145/1753326.1753478}.

\bibitem[Lau et~al.(2003)Lau, Wolfman, Domingos, and Weld]{smartedit}
T.~Lau, S.~A. Wolfman, P.~Domingos, and D.~S. Weld.
\newblock Programming by demonstration using version space algebra.
\newblock \emph{Mach. Learn.}, 53\penalty0 (1-2):\penalty0 111--156, Oct. 2003.
\newblock ISSN 0885-6125.
\newblock \doi{10.1023/A:1025671410623}.
\newblock URL \url{http://dx.doi.org/10.1023/A:1025671410623}.

\bibitem[Le and Gulwani(2014)]{flashextract}
V.~Le and S.~Gulwani.
\newblock Flashextract: A framework for data extraction by examples.
\newblock In \emph{Proceedings of the 35th ACM SIGPLAN Conference on
  Programming Language Design and Implementation}, PLDI '14, pages 542--553,
  New York, NY, USA, 2014. ACM.
\newblock ISBN 978-1-4503-2784-8.
\newblock \doi{10.1145/2594291.2594333}.
\newblock URL \url{http://doi.acm.org/10.1145/2594291.2594333}.

\bibitem[Le~Goues et~al.(2012)Le~Goues, Dewey-Vogt, Forrest, and
  Weimer]{LeGDewForWei12}
C.~Le~Goues, M.~Dewey-Vogt, S.~Forrest, and W.~Weimer.
\newblock A systematic study of automated program repair: Fixing 55 out of 105
  bugs for \$8 each.
\newblock In \emph{Proceedings of the 34th International Conference on Software
  Engineering}, ICSE '12, pages 3--13, Piscataway, NJ, USA, 2012. IEEE Press.
\newblock ISBN 978-1-4673-1067-3.
\newblock URL \url{http://dl.acm.org/citation.cfm?id=2337223.2337225}.

\bibitem[Leung et~al.(2015)Leung, Sarracino, and Lerner]{LeungSL15}
A.~Leung, J.~Sarracino, and S.~Lerner.
\newblock Interactive parser synthesis by example.
\newblock In \emph{Proceedings of the 36th {ACM} {SIGPLAN} Conference on
  Programming Language Design and Implementation, Portland, OR, USA, June
  15-17, 2015}, pages 565--574, 2015.
\newblock \doi{10.1145/2737924.2738002}.
\newblock URL \url{http://doi.acm.org/10.1145/2737924.2738002}.

\bibitem[Lieberman(2001)]{lieberman2001your}
H.~Lieberman.
\newblock \emph{Your wish is my command: Programming by example}.
\newblock Morgan Kaufmann, 2001.

\bibitem[Long and Rinard(2016)]{prophet}
F.~Long and M.~Rinard.
\newblock Automatic patch generation by learning correct code.
\newblock In \emph{Proceedings of the 43rd Annual ACM SIGPLAN-SIGACT Symposium
  on Principles of Programming Languages}, POPL 2016, pages 298--312, 2016.

\bibitem[Mitchell(1982)]{Mitchell82}
T.~M. Mitchell.
\newblock Generalization as search.
\newblock \emph{Artif. Intell.}, 18\penalty0 (2), 1982.

\bibitem[Osera and Zdancewic(2015)]{OseraZ15}
P.~Osera and S.~Zdancewic.
\newblock Type-and-example-directed program synthesis.
\newblock In \emph{Proceedings of the 36th {ACM} {SIGPLAN} Conference on
  Programming Language Design and Implementation, Portland, OR, USA, June
  15-17, 2015}, pages 619--630, 2015.
\newblock \doi{10.1145/2737924.2738007}.
\newblock URL \url{http://doi.acm.org/10.1145/2737924.2738007}.

\bibitem[Piatetsky-Shapiro(1991)]{PiaSha91}
G.~Piatetsky-Shapiro.
\newblock Discovery, analysis and presentation of strong rules.
\newblock In G.~Piatetsky-Shapiro and W.~J. Frawley, editors, \emph{Knowledge
  Discovery in Databases}, pages 229--248. AAAI Press, 1991.

\bibitem[Raychev et~al.(2013)Raychev, Sch{\"{a}}fer, Sridharan, and
  Vechev]{RaychevSSV13}
V.~Raychev, M.~Sch{\"{a}}fer, M.~Sridharan, and M.~T. Vechev.
\newblock Refactoring with synthesis.
\newblock In \emph{Proceedings of the 2013 {ACM} {SIGPLAN} International
  Conference on Object Oriented Programming Systems Languages {\&}
  Applications, {OOPSLA} 2013, part of {SPLASH} 2013, Indianapolis, IN, USA,
  October 26-31, 2013}, pages 339--354, 2013.
\newblock \doi{10.1145/2509136.2509544}.
\newblock URL \url{http://doi.acm.org/10.1145/2509136.2509544}.

\bibitem[Raychev et~al.(2014)Raychev, Vechev, and Yahav]{RayVecYah14}
V.~Raychev, M.~Vechev, and E.~Yahav.
\newblock Code completion with statistical language models.
\newblock In \emph{Proceedings of the 35th ACM SIGPLAN Conference on
  Programming Language Design and Implementation}, PLDI '14, pages 419--428,
  New York, NY, USA, 2014. ACM.
\newblock ISBN 978-1-4503-2784-8.
\newblock \doi{10.1145/2594291.2594321}.
\newblock URL \url{http://doi.acm.org/10.1145/2594291.2594321}.

\bibitem[Raychev et~al.(2015)Raychev, Vechev, and Krause]{RaychevVK15}
V.~Raychev, M.~T. Vechev, and A.~Krause.
\newblock Predicting program properties from "big code".
\newblock In \emph{Proceedings of the 42nd Annual {ACM} {SIGPLAN-SIGACT}
  Symposium on Principles of Programming Languages, {POPL} 2015, Mumbai, India,
  January 15-17, 2015}, pages 111--124, 2015.
\newblock \doi{10.1145/2676726.2677009}.
\newblock URL \url{http://doi.acm.org/10.1145/2676726.2677009}.

\bibitem[Schkufza et~al.(2013)Schkufza, Sharma, and Aiken]{Schkufza0A13}
E.~Schkufza, R.~Sharma, and A.~Aiken.
\newblock Stochastic superoptimization.
\newblock In \emph{Architectural Support for Programming Languages and
  Operating Systems, {ASPLOS} '13, Houston, TX, {USA} - March 16 - 20, 2013},
  pages 305--316, 2013.
\newblock \doi{10.1145/2451116.2451150}.
\newblock URL \url{http://doi.acm.org/10.1145/2451116.2451150}.

\bibitem[Singh and Gulwani(2012{\natexlab{a}})]{cav12}
R.~Singh and S.~Gulwani.
\newblock Synthesizing number transformations from input-output examples.
\newblock In \emph{Proceedings of the 24th International Conference on Computer
  Aided Verification}, CAV'12, pages 634--651, Berlin, Heidelberg,
  2012{\natexlab{a}}. Springer-Verlag.
\newblock ISBN 978-3-642-31423-0.
\newblock \doi{10.1007/978-3-642-31424-7_44}.
\newblock URL \url{http://dx.doi.org/10.1007/978-3-642-31424-7_44}.

\bibitem[Singh and Gulwani(2012{\natexlab{b}})]{vldb12}
R.~Singh and S.~Gulwani.
\newblock Learning semantic string transformations from examples.
\newblock \emph{Proc. VLDB Endow.}, 5\penalty0 (8):\penalty0 740--751, Apr.
  2012{\natexlab{b}}.
\newblock ISSN 2150-8097.
\newblock \doi{10.14778/2212351.2212356}.
\newblock URL \url{http://dx.doi.org/10.14778/2212351.2212356}.

\bibitem[Singh and Solar-Lezama(2011)]{storyboardfse}
R.~Singh and A.~Solar-Lezama.
\newblock Synthesizing data structure manipulations from storyboards.
\newblock In \emph{SIGSOFT FSE}, pages 289--299, 2011.

\bibitem[Solar-Lezama(2008)]{sketchthesis}
A.~Solar-Lezama.
\newblock \emph{Program Synthesis By Sketching}.
\newblock PhD thesis, EECS Dept., UC Berkeley, 2008.

\bibitem[Solar-Lezama et~al.(2005)Solar-Lezama, Rabbah, Bodik, and
  Ebcioglu]{solar2005bitstreaming}
A.~Solar-Lezama, R.~Rabbah, R.~Bodik, and K.~Ebcioglu.
\newblock Programming by sketching for bit-streaming programs.
\newblock In \emph{PLDI}, 2005.

\bibitem[Udupa et~al.(2013)Udupa, Raghavan, Deshmukh, Mador{-}Haim, Martin, and
  Alur]{UdupaRDMMA13}
A.~Udupa, A.~Raghavan, J.~V. Deshmukh, S.~Mador{-}Haim, M.~M.~K. Martin, and
  R.~Alur.
\newblock {TRANSIT:} specifying protocols with concolic snippets.
\newblock In \emph{{ACM} {SIGPLAN} Conference on Programming Language Design
  and Implementation, {PLDI} '13, Seattle, WA, USA, June 16-19, 2013}, pages
  287--296, 2013.
\newblock \doi{10.1145/2462156.2462174}.
\newblock URL \url{http://doi.acm.org/10.1145/2462156.2462174}.

\bibitem[Yakout et~al.(2011)Yakout, Elmagarmid, Neville, Ouzzani, and
  Ilyas]{YakElmNevOuzIly11}
M.~Yakout, A.~K. Elmagarmid, J.~Neville, M.~Ouzzani, and I.~F. Ilyas.
\newblock Guided data repair.
\newblock \emph{Proc. VLDB Endow.}, 4\penalty0 (5):\penalty0 279--289, Feb.
  2011.
\newblock ISSN 2150-8097.
\newblock \doi{10.14778/1952376.1952378}.
\newblock URL \url{http://dx.doi.org/10.14778/1952376.1952378}.

\bibitem[Yuan et~al.(2014)Yuan, Alur, and Loo]{YuanAL14}
Y.~Yuan, R.~Alur, and B.~T. Loo.
\newblock Netegg: Programming network policies by examples.
\newblock In \emph{Proceedings of the 13th {ACM} Workshop on Hot Topics in
  Networks, HotNets-XIII, Los Angeles, CA, USA, October 27-28, 2014}, pages
  20:1--20:7, 2014.
\newblock \doi{10.1145/2670518.2673879}.
\newblock URL \url{http://doi.acm.org/10.1145/2670518.2673879}.

\end{thebibliography}
